\documentclass[10pt,journal,twocolumn]{IEEEtran}
\usepackage{cite,array}
\usepackage{amsmath,amssymb,amsfonts}
\usepackage{algorithmic}
\usepackage{textcomp}
\usepackage[table]{xcolor}
\usepackage{makecell}
\usepackage{psfrag}
\usepackage{arydshln}
\usepackage{url}
\usepackage{soul}
\usepackage{tabularx}
\usepackage{comment}
\usepackage{graphicx,color}
\usepackage{epstopdf}
\usepackage[nolist]{acronym}
\usepackage{multirow}
\usepackage{longtable}
\usepackage{adjustbox}

\usepackage{etoolbox}
\usepackage{wrapfig}
\usepackage{epsfig}
\usepackage{color}
\usepackage{amsmath}
\usepackage{amssymb}
\usepackage{amsthm}
\newtheorem{theorem}{Theorem}[]
\newtheorem{lemma}[]{Lemma}
\newtheorem{corollary}{Corollary}[theorem]
\usepackage{mathabx}
\usepackage{enumerate}
\usepackage{multirow}
\usepackage{bbm}
\usepackage{algorithm}
\usepackage{algorithmic}
\usepackage{lipsum,amsmath,multicol}
\usepackage{stfloats}
\usepackage{arydshln} 
\usepackage{amsfonts}
\usepackage{dsfont}
\usepackage{upgreek}
\usepackage{tikz}
\usepackage{pgfplots}
\usepackage{pgfplotstable}
\pgfplotsset{compat=newest}
 \usetikzlibrary{plotmarks}
  \usetikzlibrary{arrows.meta}
  \usepgfplotslibrary{patchplots}
   \usepackage{subcaption}
\usepackage{verbatim}
\usepackage{booktabs} 
\usepackage{stfloats}
 
\usepackage[capitalise]{cleveref}
\Crefname{equation}{Eq.\!}{Eqs.\!}
\Crefname{figure}{Fig.\!}{Figs.\!}
\Crefname{tabular}{Tab.\!}{Tabs.\!}
\Crefname{section}{Section\!}{Sections.\!}

\definecolor{Linen}{rgb}{0.9803,0.9411,0.9019}
\definecolor{White}{rgb}{1,1,1}
\definecolor{Lightred}{rgb}{1,0.3803,0.3803}
\definecolor{Coral}{rgb}{1,0.4980,0.3137}
\definecolor{Grayblue}{rgb}{0.9411,0.9411,0.9803}
\definecolor{DarkLinen}{rgb}{0.729,0.7176,0.635}
\begin{document}
\begin{acronym}
\acro{NCCS}{Non-central chi-square}
\acro{RVs}{random variables}
\acro{dof}{degrees of freedom}
\acro{CDF}{cumulative distribution function}
\acro{EVT}{Extreme value theory}
\acro{i.i.d.}{independent and identically distributed}
\acro{i.n.i.d.}{independent and non-identically distributed}
\acro{SNR}{signal-to-noise ratio }
\acro{PDF}{probability density function }
\acro{MIMO}{ multiple-input multiple-output}
\acro{RIS}{reconfigurable intelligent surfaces}
\acro{OFDM}{Orthogonal Frequency-Division Multiplexing}
\acro{FSO}{free-space optical communication}
\acro{MISO}{Multiple-input and single-output}
\acro{SC}{selection combining}
\acro{URLLC}{Ultra-Reliable and Low-Latency Communication}
\acro{SWIPT}{simultaneous wireless information and power transfer}
\acro{SINR}{signal-to-interference-plus-noise ratio }
\acro{SIR}{signal-to-interference-ratio}
\acro{UAVs}{unmanned aerial vehicles}
\acro{TAS}{transmit antenna selection}
\end{acronym}
\bstctlcite{IEEEexample:BSTcontrol}
\title{Distribution of $k$-th Maximum Order Statistics of Independent, \& Non-Identical SNR random variables in $\kappa-\mu$ fading
and its applications in $6G$ }
	\author{ Srinivas Sagar, Athira Subhash, and  Sheetal Kalyani \\
      \thanks{\hspace{-0.7cm} \\
Srinivas Sagar and Sheetal Kalyani are with the Dept. of Electrical Engg., IIT Madras, India. Emails: \{ee21d051@smail, 
skalyani@ee\}.iitm.ac.in.\\Athira Subash is with the  Email:athira3003@gmail.com.}}

	\maketitle
	\begin{abstract}
		\textcolor{black} {This paper employs extreme value theory to establish the asymptotic distribution of the $k$-th maximum order statistics of signal to noise ratio (SNR)  for a $\kappa-\mu$ fading channel with independent and non-identically distributed (i.n.i.d.) parameters. 
        Since $\kappa-\mu$ encompasses well-known distributions such as Rice, Rayleigh, and Nakagami-m, the order statistics for these are also derived as special cases. We demonstrate the practical significance of our results by showcasing their applicability in several applications, including antenna selection in MIMO systems, backscatter systems, reconfigurable intelligent surfaces, and UAV-assisted relay selection systems. Furthermore, by utilizing the $k$-th maximum order statistics, we derive expressions for outage probability and average throughput of $k$-th maximum order statistics. Comprehensive Monte Carlo simulations are carried out to verify the accuracy of the proposed results.
 }
	\end{abstract}
 
	
	\section{Introduction}
    Researchers have developed various statistical distributions to model fading in wireless channel accurately. Among the different fading models,  the $\kappa-\mu$ fading distribution     \cite{yacoub2007kappa,cotton2008kappa,cotton2009channel,peppas2011sum,kumar2015coverage},
    a generalized fading model has gained significant attention. The $\kappa-\mu$ fading distribution encompasses a broad range of fading distributions observed in real-world contexts, such as Rayleigh and Rice distributions. The $\kappa-\mu$ fading distribution has found widespread applications in emerging wireless and $6G$ communication systems,, including body-to-body communications \cite{cotton2008kappa,cotton2009channel}, performance analysis \cite{peppas2011sum,kumar2015coverage,bhargav2017co}, cognitive radio \cite{sofotasios2012energy}, antenna diversity \cite{da2009accurate}, \ac{MIMO} \cite{zhang2013effective}, secrecy \cite{bhargav2016secrecy, moualeu2017secrecy}, and \ac{RIS} \cite{charishma2021outage} systems.  
    \par The authors of \cite{cotton2008kappa} utilized $\kappa-\mu$ fading to characterize channels in body-to-body communications, and \cite{cotton2009channel} extended these results to multi-antenna bodyworn systems. The distribution of the sum of \ac{i.i.d.} \cite{da2009accurate} and \ac{i.n.i.d.} \cite{peppas2011sum} $\kappa-\mu$ \ac{RVs} is presented, with applications in diversity receivers.
    Performance analyses have been conducted for $\kappa-\mu$ fading channels, covering aspects like coverage probability and outage probability in interference-limited scenarios, as demonstrated in \cite{kumar2015coverage, bhargav2017co}. Furthermore, \cite{charishma2021outage} derived outage probability expressions over $\kappa-\mu$ fading channels in \ac{RIS}-aided communication systems. Additionally, \cite{zhang2013effective} provided effective throughput expressions for multiple input single output systems under both \ac{i.i.d.} and \ac{i.n.i.d.} $\kappa-\mu$ fading channels. Authors of \cite{sofotasios2012energy} investigated the performance of energy detection in cognitive radio systems under $\kappa-\mu$ fading. Moreover, \cite{bhargav2016secrecy, moualeu2017secrecy} provided performance analyses regarding secrecy capacity and secrecy outage probability for $\kappa-\mu$ fading channels. Besides the mentioned applications, various other systems, such as \ac{FSO} \cite{gupta2018performance}, relay systems \cite{fikadu2015outage}, and \ac{OFDM} \cite{lafci2022performance}, have also utilized $\kappa-\mu$ fading in their performance analyses.   
    \par When examining the performance of wireless communication systems, such as transmit antenna selection, secrecy outage probability with multiple eavesdroppers, and selection combining (SC), there is often a need to determine the $k$-th maximum order statistics of \ac{SNR} or \ac{SIR} \ac{RVs} for the considered channel fading. However, deriving the exact order statistics for complicated channel fading, like $\kappa-\mu$, can be extremely challenging and may not be feasible. Additionally, deriving the order statistics for \ac{i.n.i.d.} $\kappa-\mu$ RVs is even more difficult. Therefore, the primary aim of this paper is to derive the $k$-th maximum order statistics of SNR RVs with $\kappa-\mu$ fading using \ac{EVT}. 
   The EVT based expressions are significantly simpler compared to the exact distributions of maxima or minima, particularly when the underlying individual distributions are themselves complex.
    Understanding the $k$-th maximum order statistics is crucial in various communication system applications, including \ac{TAS} \cite{park2009outage,duan2019asymptotic}, relay systems \cite{xia2013spectrum,subhash2021cooperative}, \ac{MIMO} systems \cite{pun2010performance}, cognitive radio networks \cite{haider2015spectral,al2019asymptotic,subhash2020transmit}, \ac{URLLC} \cite{mehrnia2022extreme}. \ac{EVT} has been extensively utilized to derive these order statistics.
   \begin{table*}[t]
        \centering
        \caption{Comparison across existing literature}
        \label{tab:refer}
        \resizebox{\textwidth}{!}{%
        \begin{tabular}{|c|p{3cm}|c|c|c|c|c|p{3cm}|}
            \hline
            \textbf{Ref} & \textbf{Random Variable} & \textbf{SNR/SIR} & \textbf{\ac{i.i.d.}} & \textbf{\ac{i.n.i.d.}} & \textbf{Max Order} & \textbf{$k$-th Max Order} & \textbf{Applications} \\
            \hline
            \cite{kalyani2006extreme} & Rayleigh & - & Yes & No & Yes & No & OFDM channel tracking \\
            \hline
            \cite{kalyani2012asymptotic} & Gamma & - & Yes & No & Yes & No & Proportional Fair scheduler (PFS) in wireless systems \\
            \hline
            \cite{al2018asymptotic} & Exponential, Weibull, Gamma, $\alpha - \mu$ and Gamma-Gamma & SNR & Yes & No & Yes & Yes & Antenna diversity \\
            \hline
            \cite{al2019asymptotic} & Gamma & SIR & Yes & No & Yes & Yes & Cognitive Radio Networks \\
            \hline
            \cite{subhash2019asymptotic} & $\kappa-\mu$ shadowed & SIR & Yes & No & Yes & No & Antenna selection in MIMO \\
            \hline
            \cite{al2020performance} & Exponential & SNR & Yes & No & Yes & Yes & Multi-Tag Backscatter Systems \\
            \hline
            \cite{subhash2020transmit} & $\kappa-\mu$ shadowed & SIR & Yes & No & No & No & Cognitive Radio Networks \\
            \hline
            \cite{subhash2021cooperative} & Rayleigh & SNR & Yes & Yes & Yes & No & Cooperative Relaying in SWIPT network \\
            \hline
            \cite{subhash2022asymptotic} & Rician & SNR & Yes & Yes & Yes & Yes & UAV assisted IoT \\
            \hline
            \cite{sagar2026multi} & \ac{NCCS} with one degree of freedom & SNR & Yes & Yes & Yes & Yes & Multi-RIS communication systems \\
            \hline
            This paper & $\kappa-\mu$ (includes Rayleigh, Rician, Nakagami, and \ac{NCCS}) & SNR & \textbf{Yes} & \textbf{Yes} & \textbf{Yes} & \textbf{Yes} & Secrecy in multicasting, antenna selection in MIMO systems, and UAV-assisted relay selection system, RIS, Backscatter systems \\
            \hline
        \end{tabular}%
       }
    \end{table*}        
    \par \ac{EVT} has been applied in various contexts to derive limiting distributions and analyze performance metrics in wireless communication systems. In \cite{park2009outage}, \ac{EVT} was utilized to derive the limiting distributions of the multicasting channel. Similarly, in systems employing \ac{TAS}, the limiting distribution of \ac{SNR} at the secondary receiver was derived using \ac{EVT} \cite{duan2019asymptotic}. Cognitive radio networks, which opportunistically utilize available spectrum, have also benefited from \ac{EVT}. In \cite{xia2013spectrum, subhash2020transmit}, \ac{EVT} was employed to derive the limiting distributions of \ac{SNR}/\ac{SIR} at the secondary network. Furthermore, authors in \cite{al2019asymptotic} derived the distribution of the $k$-th maximum \ac{SIR}, while \cite{haider2015spectral} derived average spectral and energy efficiency in cognitive radio networks. In the realm of \ac{MIMO} systems, \ac{EVT} was used in \cite{pun2010performance} to derive the limiting distributions of effective \ac{SIR} and effective \ac{SINR}, and to present the sum-rate and scaling laws of the system. In all these works, the predominant focus has been on \ac{i.i.d.} \ac{RVs}. The characterization of the asymptotic distribution of order statistics for \ac{i.n.i.d.} \ac{RVs} is somewhat challenging compared to that of \ac{i.i.d.} \ac{RVs}. Authors of \cite{subhash2021cooperative} initiated the characterization of \ac{i.n.i.d.} \ac{RVs} in communication literature, with cooperative relaying in \ac{SWIPT} network and further extended for Rician fading \ac{RVs} in \cite{subhash2022asymptotic}. The authors in \cite{sagar2026multi} studied the performance of multi-\ac{RIS} communication systems, focusing on selecting the \ac{RIS} link with the highest SNR at the destination using EVT. Table \ref{tab:refer} provides a comparative analysis of the existing literature.
   
    To the best of our knowledge, none of the existing works have investigated the $k$-th maximum order statistics of the $\kappa-\mu$ fading distribution. Motivated by this, we employ \ac{EVT} to characterize the asymptotic distribution of $k$-th maximum \ac{i.n.i.d.} $\kappa-\mu$ \ac{RVs}. We highlight several applications where the obtained results can be utilized.
    \par Secure communication aims to facilitate the accurate retrieval of source information by the intended recipient while minimizing the exposure of this information to eavesdroppers, ensuring confidentiality. Secrecy capacity refers to the maximum amount of secret information that can be reliably transmitted over a communication channel in the presence of an eavesdropper. One must analyze the maximum order statistics of the fading channel under consideration to determine performance metrics such as secrecy outage and capacity. For instance, authors in \cite{badrudduza2019enhancing} examined a wireless multicast network operating over a $\kappa-\mu$ fading channel with multiple eavesdroppers. Calculating the capacity of the eavesdroppers' channel necessitates analyzing the maximum order statistics of i.i.d $\kappa-\mu$ \ac{RVs}. 
    Therefore, the proposed \ac{i.n.i.d.} framework can also be used to characterize the performance of such \ac{i.i.d.} systems as a special case.
    \par Another potential area is in analyzing the performance of MIMO systems employing antenna selection. In various wireless contexts, employing optimal antenna selection in MIMO systems delivers considerable advantages in efficiency and reliability. Additionally, one can incorporate Maximum Ratio Combining (MRC) at the receiver alongside antenna selection at the transmitter \cite{kumbhani2016performance}. In each of these scenarios, determining the optimal link involves analyzing the $k$-th maximum order statistics. Thus, our results can be utilized to assess the performance of these systems.\\
    In \cite{cotton2009channel}, the authors found that the $\kappa-\mu$ distribution accurately models small-scale fading in body-to-body channels. They used SC for multiple-antenna body-worn systems. Our findings apply to such scenarios for performance characterization.
    \par As $\kappa-\mu$ fading encompasses special cases like Rice, Rayleigh, and Nakagami-m, we also discuss applications involving these specific instances. Authors of  \cite{ji2020secrecy} investigated the two-hop cognitive secrecy transmission scheme utilizing decoding and forwarding (DF) \ac{UAVs} as relays with energy harvesting in Nakagami-m channels. Notably, the terminal node is equipped with multiple antennas, and best antenna selection is considered for signal reception. Statistics for wireless systems employing DF opportunistic relays (ORs) under outdated channel state information (CSI) and i.i.d. generalized-Rician fading environments are derived in \cite{le2019opportunistic} using a \ac{NCCS} RV distribution and SC. Therefore, our derived results can effectively characterize system performance. \\
    \\
    In summary, the main contributions of this work are as follows:
    \begin{enumerate}
      \item We utilize EVT to obtain the asymptotic distribution of the $k$-th maximum order statistics from a sequence of \ac{i.n.i.d.} \ac{SNR} \ac{RVs} with $\kappa-\mu$ fading, with normalizing constants $a_N$ and $b_N$.  
      \item  Further, we derive the asymptotic distribution of the $k$-th maximum order statistics from a sequence of \ac{i.i.d.} \ac{SNR}  \ac{RVs} with $\kappa-\mu$ fading as a specific case, assuming all parameters to be identical.
      \item  Given that $\kappa-\mu$ fading encompasses instances such as Rice, Rayleigh, and Nakagami-m, we derive the $k$-th maximum order statistics for each of these special cases. Assuming $\kappa$ to be i.n.i.d. and $\mu$ to be i.i.d. covers distributions like Rice and Rayleigh while assuming $\kappa$ to be i.i.d. and $\mu$ to be i.n.i.d. covers the Nakagami-m distribution.
      \item Furthermore, we utilize the derived results to formulate expressions for outage probability and average throughput, crucial for analyzing system performance. We validated these results through extensive simulations.
      \item Lastly, we demonstrated the practical utility of our derived results across various applications, including backscatter systems,  RISs, MIMO systems, and UAVs with relays.
    \end{enumerate}
    \par The organization of the paper is as follows. Section II gives an overview of the $\kappa-\mu$ fading model. Section III presents the results of the asymptotic distribution of $k$-th maximum order statistics for \ac{i.n.i.d.} \ac{SNR} \ac{RVs} in $\kappa-\mu$ fading channels, along with the derivation of expressions for average throughput and outage capacity. Section IV conducts an analysis of special cases of the $\kappa-\mu$ distribution. Section V explores various applications of the derived results in  communications domain. Section VI provides  simulation results supporting our theoretical analysis, Section VII concludes the paper.
 
    In this paper, we employ the following notation: $f_{X}\left ( . \right )$ and $F_{X}\left ( . \right )$ denote the probability density function and cumulative distribution functions of a RV $X$, respectively. The expectation of $X$ is denoted as $\mathbb{E}\left ( X \right )$, and $\mathbb{P}\left ( A \right )$ denotes the probability of an event $A$.
   \section{$\kappa-\mu$ distribution}
    The received signal envelope, labeled as R, can be accurately portrayed with respect to the in-phase and quadrature components of the fading signal \cite[eq. (6)]{yacoub2007kappa}
    \begin{equation}
        R^{2}= \sum_{i=1}^{\mu}\left ( X_{i}+p_i \right )^{2}+\sum_{i=1}^{\mu}\left ( Y_{i}+q_i \right )^{2}.
    \end{equation}
   Let $R_i^2=\left ( X_{i}+p_i \right )^{2}+\left ( Y_{i}+q_i \right )^{2}$, so that $R^2=\sum_{i=1}^{\mu}R_i^2$. In this context,  $\mu$ signifies the number of multipath clusters, $X_i$ and  $Y_i$ are distinct Gaussian random processes with means  $\mathbb{E}\left ( X_{i} \right )=\mathbb{E}\left ( Y_{i} \right )=0$ 
    and variances $\mathbb{E}\left ( X_{i}^{2} \right )=\mathbb{E}\left ( Y_{i}^{2} \right )=\sigma^{2}$, $p_i$ and $q_i$ denote the mean values of the in-phase and quadrature-phase components of the $i^{th}$ multipath cluster. Now, suppose $\gamma=R^{2}$ is the instantaneous \ac{SNR} of a $\kappa-\mu$ fading signal, the \ac{PDF} of $\gamma$ can be written as \cite{yacoub2007kappa}
    \begin{align}
        f_{\gamma}\left ( \gamma  \right )=&\frac{\mu \left ( 1+\kappa \right )^{\frac{\mu +1}{2}}\gamma ^{\frac{\mu -1}{2}}}{\kappa^{\frac{\mu -1}{2}}\bar{\gamma}^{\frac{\mu +1}{2}}\exp \left ( \mu \kappa \right )}\exp \left ( -\frac{\mu\left ( 1+\kappa \right )\gamma}{\bar{\gamma}} \right ) \nonumber \\
        &\mathit{I_{\mu-1}}\left ( 2\mu\sqrt{\frac{\kappa(1+\kappa)\gamma}{\bar{\gamma}}} \right ).
    \end{align}
    Here, $\mathit{I_{\nu}}$ is the modified Bessel function of the first kind with order $\nu$. Moreover, $\bar{\gamma}=\mathbb{E}\left ( \gamma \right )$, while $\kappa$ and $\mu$ are fading parameters, where $\kappa=\frac{\sum_{i=1}^{\mu}p_i^2+q_i^2}{2\mu\sigma^2}$. The \ac{CDF} of $\gamma$ is given by
    \begin{equation}\label{eq1}
        F_{\gamma}(\gamma)=1-Q_{\mu}\left(\sqrt{2\kappa\mu}, \sqrt{\frac{2(1+\kappa)\mu\gamma}{\bar{\gamma}}}\right),
    \end{equation}
    where $Q_{\mu}(. , .)$ is the Marcum-Q function \cite{molisch2012wireless}. Next, we derive the $k$-th maximum order statistics distribution of a sequence of \ac{i.n.i.d.} SNR ($\gamma$) \ac{RVs}.
    \section{$k$-th maximum order statistics of i.n.i.d. $\kappa-\mu$  RVs  } \label{sec_kth_max}
     Consider a sequence $\{\gamma_{1},\gamma_{2},..,\gamma_{N} \}$, where each element represents an SNR RV, denoted as $\gamma_{n}$. These RVs follow a \ac{CDF} $F_{\gamma_{n}}(\gamma)$ as defined in equation (\ref{eq1}), where $n$ ranges from $1$ to $N$. Let $\gamma_{ (1: N)} \leq \gamma_{(2: N)} \leq \cdots \leq \gamma_{(N:N)}$, be the order statistics where the $k$-th order statistic is given by $\gamma_{(N-k+1:N)}$. Here, $k=1$ represents the maximum order statistics i.e. $\gamma_{(N:N)}=\gamma_{max}^{N}=\max \left \{\gamma_{1},\gamma_{2},..,\gamma_{N} \right \}$. Finding the exact \ac{CDF} of $k$-th order statistic $\gamma_{(N-k+1:N)}$  involves very complicated expression as given  \cite[(5.2.1)]{david2003order} 
    \begin{align}
        & F_{\gamma_{(N-k+1:N)}}(\gamma)=\sum_{i=k}^{N} \sum_{S_{i}} \prod_{n=1}^{i} F_{\gamma_{j_n}}(\gamma) \nonumber\\ & \times  \prod_{n=i+1}^{N}\left[1-F_{\gamma_{j_n}}(\gamma)\right], \ k=1,2,\cdots,N,
        \label{cdf_kth_max_exact}
    \end{align}
    where the summation $S_{i}$ is over all the permutations $\left(j_{1}, \ldots, j_{N}\right)$ of $1, \ldots, N$ for which $j_{1}<\cdots<j_{i}$ and $j_{i+1}<\cdots<j_{N}$. A much simpler asymptotic expression 
    of $\gamma_{(N-k+1:N)}$ can be computed with the help of \ac{EVT}.  
    Asymptotic order statistics can be derived with the help of the results presented in \cite{barakat2002limit}. For ease of understanding, we will repeat the important Lemma here. 
 
    \begin{lemma} \label{thm_order_stat} \color{black}
    Assume that for suitable normalizing constants $a_{N}>0,$ $b_{N}$
    \begin{equation}
        \delta_{N}=\max _{1 \leq n \leq N} 1-F_{\gamma_n} \left(a_N\gamma+b_N\right) \rightarrow 0 \text { as } N \rightarrow \infty.
        \label{ua1}
    \end{equation}
    Then $\tilde{\phi}_{N-k+1:N}(\gamma) = \mathbb{P}\left(\frac{\gamma_{(N-k+1:N)}-b_N}{a_N} \leq \gamma \right)$ converges weakly to a non degenerate distribution function $\tilde{\phi}_{N-k+1}\left(\gamma \right)$ if and only if, for all $\gamma$ for which $\tilde{\phi}_{N-k+1}\left(\gamma \right)>0$, the limit
    \begin{equation}
        \tilde{u}(\gamma) = \lim _{N \rightarrow \infty} \sum_{n=1}^{N} 1-F_{\gamma_n}(a_N\gamma+b_N) \ \text{is finite,}
        \label{ua2}
    \end{equation}
    and the function
    \begin{equation}
        \tilde{\phi}_{k}\left(\gamma \right) = \sum\limits_{n=0}^{k-1} \frac{\tilde{u}^n(\gamma)}{n!} \exp (-\tilde{u}(\gamma)), 
        \label{ua-3}
    \end{equation}
    is a non degenerate distribution. 
    \end{lemma}
    	
    \begin{proof}
    Please refer \cite{barakat2002limit} for the detailed proof.
    \end{proof}     
   The distribution of $k$-th maximum order statistics of  \ac{i.n.i.d.} \ac{SNR}  \ac{RVs} with
    $\kappa-\mu$ fading is derived by identifying the normalizing constants $a_{N}$ and $b_{N}$ that ensure the limit in (\ref{ua2}) is finite. We now derive the normalized $k$-th maximum order statistics $\tilde{\gamma }_{\left ( N-k+1:N \right )}=\frac{\gamma _{\left ( N-k+1:N \right )}-b_{N}}{a_{N}}$ for normalizing constants $a_{N}$ and $b_{N}$ in the following theorem. 
   \begin{theorem} \label{thm_kth_max}   
        The asymptotic distribution of normalized $k$-th maximum $\tilde{\gamma }_{\left ( N-k+1:N \right )}$ of a sequence of \ac{i.n.i.d.} \ac{SNR}  \ac{RVs} in  $\kappa-\mu$ fading,  as $N \to \infty$ is given by  		
		\begin{equation}
    		F_{\tilde{\gamma}_{(N-k+1:N)}}(\gamma) = \frac{\Gamma(k,  \tilde{u}(\gamma))}{\Gamma(k)}=\frac{\Gamma(k, \exp(- \gamma))}{\Gamma(k)},
    		\label{k_th_max_cdf}
		\end{equation}
		for the choice of normalizing constants
        $$
        \begin{aligned}
            a_{N} &=\frac{\bar{\gamma}}{(1+\tilde{\kappa})\tilde{\mu}} \nonumber \\
            b_{N} &=a_N c_N \nonumber \\=&a_N \left(\log (\tilde{N})-c_0  \log (\log (\tilde{N}))+c_1 \sqrt{\log (\tilde{N})}-c_2\right),
        \end{aligned}
        $$
        Here, $c_{0}=\frac{-(2\tilde{\mu}-3)}{4}$, $c_1=2 \sqrt{\frac{\tilde{\kappa}\tilde{\mu}}{\tilde{\theta}}}$, and  $c_2=\frac{1}{\tilde{\theta}}\left[ \log\left( \frac{\sqrt{\tilde{\kappa}}^{\tilde{\mu}-\frac{1}{2}}}{\sqrt{1+\tilde{\kappa}}^{\tilde{\mu}-\frac{3}{2}}}\frac{2\sqrt{\pi\tilde{\mu}}}{\left(\frac{a_N}{\bar{\gamma}}\right)^{\frac{2\tilde{\mu}-3}{4}}} \right)-\tilde{\kappa}\tilde{\mu} \right]$. Each $\gamma_n$ follows $\kappa-\mu$ distribution with parameters $\left(\kappa_n, \mu_n\right)$.
        Assume $\left(\kappa_n, \mu_n\right)$ takes a finite set of values i.e. $\kappa_n \in\left\{\kappa_1, \kappa_2 \ldots \kappa_P\right\}$ for all $n \in\{1, \ldots N\}$ and $\mu_n \in\left\{\mu_1, \mu_2 \ldots \mu_L\right\}$  for all  $n \in\{1, \cdots N\}$ with $P, L$ as finite numbers. Choose $\tilde{\mu}$ to be the smallest among $\left\{\mu_{1},\cdots,\mu_{L}\right\}$ and $\tilde{\kappa}$ to be the smallest among $\left\{\kappa_{1},\cdots,\kappa_{P}\right\}$.  
		Also, 
        $N_{ij}=\sum_{n=1}^N \mathbb{I}_{(\kappa_n=\kappa_i, \mu_n=\mu_j)} \quad  1 \leq i \leq P , 1 \leq j \leq L$ and $\tilde{N} = \sum\limits_{n=1}^N \mathbb{I}_{(\kappa_n=\tilde{\kappa},\mu_n=\tilde{\mu})}$. Here $N_{ij}$ represents the number of times  pair $\left(\kappa_i, \mu_j\right)$ occurs among  $N$ values, and we assume $\tilde{N}\to \infty $, as $N \to \infty $. Consider $\theta_{ij}=\frac{\left ( 1+\kappa_i \right )\mu_j}{\left ( 1+\tilde{\kappa} \right )\tilde{\mu}}$, and $\tilde{\theta}=\min \left \{ \theta_{ij} \right \}_{(i,j)=(1,1)}^{(P, L)}$.\\
	\end{theorem}
	\begin{proof}
    Please refer Appendix \ref{proof_main} for the proof.
	\end{proof}
   Note that for all practical purposes, one can use 
   \begin{equation*}
         \tilde{u}(\gamma)= \sum_{(i,j) \rightarrow (1,1)}^{(P \times L) }\exp\left (-\theta_{ij}\gamma\right ) P_{ij},
    \end{equation*}
    where $P_{ij}$ is given by,
    \begin{multline*}
        P_{ij}=\frac{N_{ij}}{(\tilde{N})^{\theta_{ij}}}\exp \left[2\theta_{ij} \left(\sqrt{\frac{\kappa_i\mu_j}{\theta_{ij}}} -\sqrt{\frac{\tilde{\kappa}\tilde{\mu}}{\tilde{\theta}}} \right)\sqrt{\log (\tilde{N})}\right]\\
        \exp \left[2\theta_{ij}\sqrt{\frac{\kappa_i\mu_j}{\theta_{ij}}}\sqrt{\frac{\tilde{\kappa}\tilde{\mu}}{\tilde{\theta}}}+\log\left(\frac{\sqrt{1+\kappa_i}^{\mu_j-\frac{3}{2}}}{\sqrt{\kappa_i}^{\mu_j-\frac{1}{2}}}\frac{1}{2\sqrt{\pi \mu_j}}\right)\right] \\ \exp \left[ \left(-\kappa_{i}\mu_j\right)+ \left ( \log \left ( \frac{a_N}{\bar{\gamma}} S \right )^{\frac{2\tilde{\mu}-3}{4}} \right ) \right] \\
        \exp \left( \theta_{ij} \frac{1}{\tilde{\theta}}\left[ \log\left( \frac{\sqrt{\tilde{\kappa}}^{\tilde{\mu}-\frac{1}{2}}}{\sqrt{1+\tilde{\kappa}}^{\tilde{\mu}-\frac{3}{2}}}\frac{2\sqrt{\pi\tilde{\mu}}}{\left(\frac{a_N}{\bar{\gamma}}\right)^{\frac{2\tilde{\mu}-3}{4}}} \right)-\tilde{\kappa}\tilde{\mu} \right] \right).
    \end{multline*}   
    The asymptotic distribution results for unnormalized $k$-th maximum order statistics are derived by substituting $\gamma$ by $\frac{\gamma-b_N}{a_N}$ in (\ref{k_th_max_cdf}). The unnormalized results are especially useful in practical regimes involving physical parameter scales.

    \subsection{Average throughput and Outage capacity}
    In this sub-section, we derive the expressions for the metrics average throughput and outage capacity using the derived \ac{CDF} of the $k$-th maximum order statistics of \ac{i.n.i.d.} \ac{SNR} \ac{RVs} in $\kappa-\mu$ fading. Consider $\gamma _{(N-k+1:N)}$ be the $k$-th maximum  \ac{SNR} RV of a sequence of \ac{i.n.i.d.}  $\kappa-\mu$ \ac{RVs}. We derive the asymptotic average throughput and outage capacity expressions given the asymptotic distribution of $\gamma _{(N-k+1:N)}$. Also, we assume \ac{SNR} RV to be of the form $\gamma_a\gamma_{(N-k+1:N)}$ where $\gamma_a$ is a constant.
    \subsubsection{Average throughput}
    We derived the asymptotic distribution of normalized $k$-th maximum order statistics of  \ac{i.n.i.d.} \ac{SNR} \ac{RVs} with $\kappa-\mu$ fading in Theorem \ref{thm_kth_max}. Using this, we can determine the average throughput at the receiver. The expression for average throughput is represented as:
    \begin{equation}\label{er_1}
        C_{(N-k+1:N)}=\mathbb{E}\left [\log_2\left ( 1+\gamma_a\gamma_{(N-k+1:N)} \right )  \right ].
    \end{equation}
   
    Given that  $ F_{\gamma_{(N-k+1:N)}}(\gamma)$ is a positive RV, the expectation in (\ref{er_1}) can be evaluated using the following expression for the first moment of a positive RV 
    \begin{equation*}
        \mathbb{E}\left[ x \right]=\int_{0}^{\infty }\left( 1-F_{X}\left( x \right) \right)dx.
    \end{equation*}
    Thus, (\ref{er_1}) can be evaluated as
    \begin{equation}\label{er_2}
       C_{(N-k+1:N)}=\int_{0}^{\infty }\left( 1- \frac{\Gamma\left( k, \exp\left( - \log_2\left ( 1+\gamma_a\gamma_{(N-k+1:N)} \right ) \right) \right)}{\Gamma(k)}  \right)d\gamma.
    \end{equation}
    Therefore, the average throughput can be computed  with the assistance of numerical integration routines.
     
    
  \subsubsection{Outage Capacity}
  The outage probability with  $\gamma_{th}$ as threshold can be derived using the asymptotic distribution of $k$-th maximum order statistics  (\ref{ua3}) as
  \begin{align}
     P_{out}= & \mathbb{P}\left ( \gamma _{a} \gamma_{(N-k+1:N)}\leq \gamma _{th}\right )=F_{\gamma _{(N-k+1:N)}}\left ( \frac{\gamma _{th}}{\gamma _{a}} \right )\nonumber \\ &=\frac{\Gamma \left ( k, \exp \left ( -\frac{ \left ( \frac{\gamma _{th}}{\gamma _{a}}-b_{N} \right )}{a_{N}} \right )\right )}{\Gamma \left ( k \right )}.
 \end{align}
 Similarly, Outage capacity \cite{gu2015rf} can be calculated as 
 \begin{equation}\label{c_out}
    C_{out}=\log _{2}\left ( 1+\gamma _{th} \right )\left ( 1-F_{\gamma _{(N-k+1:N)}}\left (\gamma _{th}  \right ) \right ).
 \end{equation}

    \subsection{Special Cases} \label{sys_model}   
     
     Theorem \ref{thm_kth_max} presents general results for the $k$-th maximum order statistics of \ac{i.n.i.d.} \ac{SNR} \ac{RVs} in $\kappa-\mu$ fading. Here, we look at various special cases.
      \subsubsection{\textbf{i.i.d. case}}
     In this sub-section, we derive the asymptotic distribution of $k$-th maximum order statistics of \ac{i.i.d.} \ac{SNR} \ac{RVs} with $ \kappa-\mu$ fading. These results can be derived as a special case of (Theorem \ref{thm_kth_max}) \ac{i.n.i.d.} $\kappa-\mu$ \ac{RVs} assuming all the parameters identical, i.e., for $\kappa_i=\kappa, \mu_j=\mu$.
      \begin{corollary} \label{iid_thm_kth_max}   
        The asymptotic distribution of normalized $k$-th maximum $\tilde{\gamma }_{\left ( N-k+1:N \right )}$ of a sequence of \ac{i.i.d.} \ac{SNR} \ac{RVs} in  $\kappa-\mu$ fading, as $N \to \infty$ is given by  		
		\begin{equation}
    		F_{\tilde{\gamma}_{(N-k+1:N)}}(\gamma) = \frac{\Gamma(k, \exp(- \gamma))}{\Gamma(k)},
    		\label{iid_k_th_max_cdf}
		\end{equation}
		for the choice of normalizing constants
       $$
        \begin{aligned}
            a_{N} &=\frac{\bar{\gamma}}{(1+\kappa)\mu} \nonumber \\
            b_{N}  &=a_N \left(\log (N)-c_0  \log (\log (N))+c_1 \sqrt{\log (N)}-c_2\right),
        \end{aligned}
        $$
        Here $c_{0}=\frac{-(2\mu-3)}{4}$, $c_1=2 \sqrt{\kappa\mu}$, and  $c_2=\left[ \log\left( \frac{\sqrt{\kappa}^{\mu-\frac{1}{2}}}{\sqrt{1+\kappa}^{\mu-\frac{3}{2}}}\frac{2\sqrt{\pi\mu}}{\left(\frac{a_N}{\bar{\gamma}}\right)^{\frac{2\mu-3}{4}}} \right)-\kappa\mu \right]$.
	\end{corollary}
    \begin{proof}
        This result can be derived by substituting $\kappa_n=\kappa$, $\mu_n=\mu$, and $\tilde{N}=N$ in Theorem \ref{thm_kth_max}.
    \end{proof}
     The asymptotic distribution results for unnormalized $k$-th maximum order statistics are derived by substituting $\gamma$ by $\frac{\gamma-b_N}{a_N}$ in (\ref{iid_k_th_max_cdf}).
     \subsubsection{\textbf{$\kappa$-\ac{i.n.i.d.} and $\mu$-\ac{i.i.d.}$(\kappa_i,\mu_j=\mu)$}}
      Distributions such as Rice, Rayleigh, and \ac{NCCS} can be obtained from $\kappa-\mu$ distribution by considering $\kappa$ to be \ac{i.n.i.d.} and $\mu$ to be of \ac{i.i.d.} 
     \begin{corollary} \label{thm_kth_max_2}   
        The asymptotic distribution of normalized $k$-th maximum $\tilde{\gamma }_{\left ( N-k+1:N \right )}$ of a sequence of \ac{i.n.i.d.} $\kappa$ and \ac{i.i.d.} $\mu$ \ac{RVs} as $N \to \infty$ is given by  		
		\begin{equation}
    		F_{\tilde{\gamma}_{(N-k+1:N)}}(\gamma) = \frac{\Gamma(k,  \tilde{u}(\gamma))}{\Gamma(k)}=\frac{\Gamma(k, \exp(- \gamma))}{\Gamma(k)},
    		\label{k_th_max_cdf_1}
		\end{equation}
		for the choice of normalizing constants
        $$
        \begin{aligned}
            a_{N} &=\frac{\bar{\gamma}}{(1+\tilde{\kappa})\mu} \nonumber \\
            b_{N} &=a_N \left(\log (\tilde{N})-c_0  \log (\log (\tilde{N}))+c_1 \sqrt{\log (\tilde{N})}-c_2\right),
        \end{aligned}
        $$
        Here, $c_{0}=\frac{-(2\mu-3)}{4}$, $c_1=2 \sqrt{\frac{\tilde{\kappa}\mu}{\tilde{\theta}}}$, and  $c_2=\frac{1}{\tilde{\theta}}\left[ \log\left( \frac{\sqrt{\tilde{\kappa}}^{\mu-\frac{1}{2}}}{\sqrt{1+\tilde{\kappa}}^{\mu-\frac{3}{2}}}\frac{2\sqrt{\pi\mu}}{\left(\frac{a_N}{\bar{\gamma}}\right)^{\frac{2\mu-3}{4}}} \right)-\tilde{\kappa}\mu \right]$. Assume $\kappa_n$ takes a finite set of values i.e. $\kappa_n \in\left\{\kappa_1, \kappa_2 \ldots \kappa_P\right\}$ for all $n \in\{1, \ldots N\}$. Choose $\tilde{\kappa}$ to be the smallest among $\left\{\kappa_{1},\cdots,\kappa_{P}\right\}$.  
		Also, 
        $N_{i}=\sum_{n=1}^N \mathbb{I}_{(\kappa_n=\kappa_i)} \quad  1 \leq i \leq P $ and $\tilde{N} = \sum\limits_{n=1}^N \mathbb{I}_{(\kappa_n=\tilde{\kappa})}$. Here, $N_{i}$ represents the number of times  $\left(\kappa_i\right)$ occurs among  $N$ values, and we assume $\tilde{N}\to \infty $, as $N \to \infty $. Consider $\theta_{i}=\frac{\left ( 1+\kappa_i \right )}{\left ( 1+\tilde{\kappa} \right )}$, and $\tilde{\theta}=\min \left \{ \theta_{i} \right \}_{i=1}^{P}$.
	\end{corollary}
	\begin{proof}
    This result can be derived by substituting $\mu_n=\mu$ in Theorem \ref{thm_kth_max}.
    \end{proof}
    \paragraph{Rice distribution:} Rice distribution can be obtained by substituting $\mu=1$ in corollary \ref{thm_kth_max_2}.
    \paragraph{\ac{NCCS} distribution:}
    As shown in \cite{peppas2011sum}, \ac{NCCS} distributed \ac{RVs} can be obtained from $\kappa-\mu$ distributed \ac{RVs}. If $Y$ is a \ac{NCCS} RV with $m$ \ac{dof} having non centrality parameter $\lambda$, then the RV $\gamma=\left ( \frac{\bar{\gamma}}{2\mu\left ( 1+\kappa \right )} \right )Y$ follows $\kappa-\mu$ fading with parameters $\left(\kappa,\mu,\bar{\gamma}\right)$. Note that non-centrality parameter $\lambda=2\kappa\mu$ and \ac{dof} $m=2\mu$. For \ac{NCCS} distribution, \ac{dof} needs to be an integer hence $\mu$ in this case will be $m/2$. Here, fixing \ac{dof} and considering \ac{i.n.i.d.} non-centrality parameters lead to \ac{i.n.i.d.} $\kappa$ and \ac{i.i.d.} $\mu$. The corresponding $k$-th maximum order statistics can be obtained from corollary \ref{thm_kth_max_2}.
    \subsubsection{\textbf{$\kappa$-\ac{i.i.d.} and $\mu$-\ac{i.n.i.d.}$(\kappa_i=\kappa,\mu_j)$}}
     The Nakagami-m distribution can be obtained by setting $\kappa$ \ac{i.i.d.} $\left(\kappa \rightarrow 0\right)$ and $\mu$ to be of \ac{i.n.i.d.}. 
     \begin{corollary} \label{thm_kth_max_4}   
        The asymptotic distribution of normalized $k$-th maximum $\tilde{\gamma }_{\left ( N-k+1:N \right )}$ of a sequence of \ac{i.n.i.d.} $\mu$ and \ac{i.i.d.} $\kappa$ \ac{RVs} as $N \to \infty$ is given by  		
		\begin{equation}
    		F_{\tilde{\gamma}_{(N-k+1:N)}}(\gamma) = \frac{\Gamma(k,  \tilde{u}(\gamma))}{\Gamma(k)}=\frac{\Gamma(k, \exp(- \gamma))}{\Gamma(k)},
    		\label{k_th_max_cdf_2}
		\end{equation}
		for the choice of normalizing constants
        $$
        \begin{aligned}
            a_{N} &=\frac{\bar{\gamma}}{(1+\kappa)\tilde{\mu}} \nonumber \\
            b_{N} &=a_N \left(\log (\tilde{N})-c_0  \log (\log (\tilde{N}))+c_1 \sqrt{\log (\tilde{N})}-c_2\right),
        \end{aligned}
        $$
        Here, $c_{0}=\frac{-(2\tilde{\mu}-3)}{4}$, $c_1=2 \sqrt{\frac{\kappa\tilde{\mu}}{\tilde{\theta}}}$, and  $c_2=\frac{1}{\tilde{\theta}}\left[ \log\left( \frac{\sqrt{\kappa}^{\tilde{\mu}-\frac{1}{2}}}{\sqrt{1+\kappa}^{\tilde{\mu}-\frac{3}{2}}}\frac{2\sqrt{\pi\tilde{\mu}}}{\left(\frac{a_N}{\bar{\gamma}}\right)^{\frac{2\tilde{\mu}-3}{4}}} \right)-\kappa\tilde{\mu} \right]$. Assume $\mu_n$ takes a finite set of values i.e. $\mu_n \in\left\{\mu_1, \mu_2 \ldots \mu_L\right\}$  for all  $n \in\{1, \cdots N\}$. Choose $\tilde{\mu}$ to be the smallest among $\left\{\mu_{1},\cdots,\mu_{L}\right\}$.  
		Also, 
        $N_{j}=\sum_{n=1}^N \mathbb{I}_{( \mu_n=\sigma_j)} \quad   1 \leq j \leq L$ and $\tilde{N} = \sum\limits_{n=1}^N \mathbb{I}_{(\mu_n=\tilde{\mu})}$. Here, $N_{j}$ represents the number of times  $\mu_j$ occurs among  $N$ values, and we assume $\tilde{N}\to \infty $, as $N \to \infty $. Consider $\theta_{j}=\frac{\mu_j}{\tilde{\mu}}$, and $\tilde{\theta}=\min \left \{ \theta_{j} \right \}_{(j=1)}^{ L}$.
	\end{corollary}
	\begin{proof}
     This result can be derived by substituting $\kappa_n=\kappa$ in Theorem \ref{thm_kth_max}.
    \end{proof}    
    \section{Applications}
    Since our results are for the  \ac{i.n.i.d.} case they efficiently capture more realistic heterogeneous environments. In all applications we recover the  \ac{i.i.d.} case as a special case, while so far all existing work for these applications is restricted to only the \ac{i.i.d.} case.
     \subsection{Applications with Gamma }
\subsubsection{Backscatter}
 An Ambient Backscatter System (AmBS) typically includes an ambient RF source (such as a TV broadcasting tower, cellular base station, or Wi-Fi hotspot), a passive backscatter tag, and a reader. In contrast to traditional radios that rely on a dedicated carrier for transmission, a passive tag functions by harvesting RF energy from an existing ambient source and embedding its information onto the reflected signal.  Because it needs very little power, AmBS is seen as a useful technology for battery-free and low-cost IoT applications such as smart cities, environmental monitoring, and wearable devices. 
    In \cite{ali2024ergodic,10845856}, authors considered an AmBS model consisting of an ambient RF source $S$, a cluster of $N$ passive tags $\{T_i\}_{i=1}^N$, and a reader $R$. Since not all tags experience the same channel conditions, the system adopts an opportunistic tag selection strategy, where only the tag with the strongest channel to the reader is allowed to communicate at a given time. This approach ensures improved reliability and better utilization of the available RF energy. Specifically, the selected tag, denoted as $i^{*}$, is determined according to
    \begin{equation*}
        i^{*}=\arg\max_{i\in \left\{ 1,\cdots,N \right\}} \left| h_{t_{i}r} \right|^{2},
    \end{equation*}    
    where $\left| h_{t_{i}r} \right|^{2}$ denotes the instantaneous channel power gain (gamma distributed RV) between tag $T_i$ and the reader $R$.
   It is well known that a Gamma RV can be treated as a special case of the derived $\kappa-\mu$ results, i.e., with $\kappa \rightarrow 0$ and $\mu$ following the \ac{i.n.i.d.} case as stated in Corollary \ref{thm_kth_max_4}. Specifically, a Gamma RV with parameters $(\alpha, \theta)$ corresponds to a $\kappa-\mu$ distribution with $\kappa \rightarrow 0$, $\mu=\alpha$, and $\bar{\gamma}=\alpha \theta$. Consequently, the derived asymptotic framework readily specializes to \ac{i.n.i.d.} Gamma random variables. As a result, a wide variety of Gamma-distributed wireless communication applications can be investigated directly as special cases of the proposed analysis.
\subsubsection{Multi-RIS ( multiple reconfigurable intelligent surfaces (RISs))}
   The design and performance aspects of networks assisted by multiple RISs have been extensively investigated in previous studies \cite{do2021multi,10041765}. In  \cite{do2021multi}, the authors introduce two multi-RIS-aided schemes, namely the exhaustive RIS-aided (ERA) and opportunistic RIS-aided (ORA) schemes, and model their end-to-end (e2e) channels as Gamma distributed.  In the ERA scheme, every RIS takes part in the transmission, whereas in the ORA scheme, only one optimal RIS is selected. The ORA approach helps conserve resources while still maintaining a good trade-off between performance and efficiency. Optimal RIS selection is done using 
    \begin{equation*}
        n^{*}=\arg\max_{n\in \left\{ 1,\cdots,N \right\}} V_{n},
    \end{equation*}
     where $V_n$ follows gamma distribution. Authors derived  $F_{M_V}(x)$ as the maximum order statistics of gamma RVs as given in \ref{ORA}.
    \begin{align}\label{ORA}
        F_{M_V}(x) =& \Pr \left( \max_{1 \leq k \leq N} V_k \leq x \right) 
        = \prod_{k=1}^{N} F_{V_k}(x) \nonumber \\  &
        \approx  \prod_{k=1}^{N} \frac{\gamma \!\left(L_k \alpha_{U_k},\, \beta_{U_k} x \right)}{\Gamma(L_k \alpha_{U_k})} .
    \end{align}    
    Further, the authors obtained the CDF of $\gamma_{ORA}$ using M-staircase approximation, which is given by: \begin{equation}\label{comp} F_{\gamma_{ORA}}\left ( x \right )\approx \sum_{m=1}^{M}\int_{\frac{m-1}{M}x}^{\frac{m}{M}x}f_{h_{0}}\left ( y \right )\int_{0}^{\frac{M-m+1}{M}x}f_{M_V}\left ( z \right )dzdy, \end{equation} where $f_{M_V}\left ( z \right )\approx \frac{\mathrm{d} }{\mathrm{d} z}\left ( \prod_{k=1}^{N} \frac{\gamma \!\left(L_k \alpha_{U_k},\, \beta_{U_k} z \right)}{\Gamma(L_k \alpha_{U_k})} \right )$. It can be observed that these expressions involve multiple integrals and derivatives, making them quite complex. While \cite{sagar2026multi} characterizes the selected RIS link \ac{SNR} using a \ac{NCCS} distribution with one degree of freedom, recent literature increasingly relies on the Gamma distribution for performance analysis. As a result, the derived asymptotic framework provides a unified analytical tool for investigating such Gamma-distributed multi-RIS systems. In contrast, the results derived in Corollary \ref{thm_kth_max_4} can be directly used to obtain the maximum order statistics of Gamma random variables in a much simpler and more tractable form, thereby providing a cleaner alternative to the above formulation.
    \subsection{Applications with Nakagami-m}
  \subsubsection{UAV-aided relay and antenna selection for cognitive networks}
  Consider the system model introduced in \cite{ji2020secrecy}, where the authors employed a novel relaying technique called UAV-assisted relay, wherein the relay nodes are mobile. The cognitive radio system model comprises a secondary user source (S), a secondary user destination (D) with $M$ antennas, and $N$ UAV-assisted relays. Meanwhile, a passive eavesdropper intercepts the relayed signal within the secondary network. Initially, node S transmits signals to relay node $R_j$, where $1 \leq j \leq N$. The UAV-assisted relay operates in Decode-and-Forward (DF) mode and utilizes a relay selection strategy to pick one of the $N$ UAVs for relaying information signals to the destination node. All channels are assumed to exhibit i.i.d. Nakagami-m fading distribution. The SNR from UAV relay node $R_j$ to destination node D and UAV relay node $R_j$ to eavesdropper can be written as 
  \begin{equation}
      \gamma_{R_j D}=\frac{P_{R_j}}{\sigma_{D_i}^2} \max_{ M} \left ( \left \| h_{R_j D_i} \right \|^2 \right ),\\ \quad
      \gamma_{R_j E}=\frac{P_{R_j}}{\sigma_{E}^2} \max_{ M} \left ( \left \| h_{R_j E} \right \|^2 \right )
  \end{equation}
  The destination node's antenna employs a maximization criterion, meaning it selects the optimal antenna for transmission. 
Given the consideration of the Nakagami-m distribution, we can express this as a specific instance of $\kappa-\mu$ fading by setting $\kappa$ to be i.i.d. with $\kappa \rightarrow 0$, and $\mu$ to be i.n.i.d. 
By applying Corollary \ref{thm_kth_max_4} with $k=1$, we can compute the maximum order statistics when all $\mu$ are the same. 
\subsection{Applications with NCCS}
  \subsubsection{Wireless systems utilizing Decode-and-Forward (DF) opportunistic relays (ORs)} Consider a wireless system similar to the one described in \cite{le2019opportunistic}, where communication from the source to the destination is facilitated through a group of $N$ DF ORs. Information is sent from the source to the best-performing relay and then relayed from this selected relay to the destination. Cooperative relay networks have received significant attention in recent literature, particularly employing SC techniques. The authors in \cite{le2019opportunistic} investigated DF-OR networks within the context of generalized-Rician fading, which is capable of representing both Line-of-Sight (LoS) and Non-Line-of-Sight (NLoS) fading scenarios. It's worth mentioning that the fading gains in a generalized-Rician fading environment follow an NCCS distribution. However, managing the intricacies of such environments can be challenging. Therefore, the authors investigated a scenario where the ORs function in i.i.d. generalized-Rician fading environments, with a specific focus on situations involving only even dof. Since the channel fading gains follow the NCCS distribution, the CDF depicted in \cite[(68)]{le2019opportunistic} can be seen as a special case of the $\kappa-\mu$ fading, where $\mu$ is an even number. For a given set of $N$ relays, the SNR of each relay is denoted by $\gamma_i$ for $i=1,\ldots,N$. The CDF representing the SNR of the best relay is $\max _{1< i< N} \gamma_i $. We can once more employ Theorem-\ref{thm_kth_max} to obtain the maximum order statistics of \ac{SNR} of the best relay. It's worth noting that the authors of \cite{le2019opportunistic} derived their results for i.i.d. RVs with even dof. In fact, Theorem-\ref{thm_kth_max} can extend these findings to i.n.i.d. RVs with any dof. Furthermore, we can extend these findings to include the $k$-th maximum order statistics.
    \subsection{Applications with $\kappa-\mu$}
  \subsubsection{MIMO systems with antenna selection}
  Consider a MIMO system similar to the one outlined in \cite{kumbhani2016performance}, featuring $N_t$ transmit antennas and $N_r$ receive antennas. In this scenario, $h_{ji}$ signifies the channel fading coefficient between the $i^{th}$ transmitting antenna and the $j^{th}$ receiving antenna. Furthermore, $\gamma_{ji} = \left | h_{ji} \right |^{2}$ represents the received SNR. 
  All channels are assumed to undergo $\kappa-\mu$ fading. The transmitter is assumed to have only one RF chain available, and it selects the antenna that maximizes the received SNR. MIMO with \ac{TAS} and maximal ratio combining (MRC) achieves enhanced diversity with reduced complexity. The authors in \cite{kumbhani2016performance} examined two types of systems: those with joint transmit and receive antenna selection and those employing \ac{TAS}/MRC systems. In joint transmit and receive antenna selection systems, only the antenna corresponding to the best link is chosen at the receiver, defined as
\begin{equation}
    \gamma_{max}=\max_{1\leq i \leq N_t , 1\leq j \leq N_r} \gamma_{ji}.
\end{equation}

In \ac{TAS}/MRC systems, MRC is utilized at the receiver, and the transmitting antenna that maximizes the SNR at the receiver is selected for communication, represented as,
\begin{equation}
    \gamma_{max}^{TAS}=\max_{1\leq i \leq N_t }    \left ( \sum_{j=1}^{N_r} \gamma_{ji} \right ).
\end{equation} 

Both systems necessitate characterizing the maximum order statistics of \ac{SNR} \ac{RVs} in $\kappa-\mu$ fading. 
By setting $k=1$ in Corollary \ref{iid_thm_kth_max}, the maximum order statistics considered in \cite{kumbhani2016performance} are readily obtained. Notably, our framework extends beyond the \ac{i.i.d.} assumption and is sufficiently general to model \ac{i.n.i.d.} $\kappa$–$\mu$ \ac{RVs}. 
\subsubsection{Indoor Body-to-Body Communications}
The authors in \cite{cotton2009channel} showed that the $\kappa-\mu$ distribution offers the most accurate description of small-scale fading in body-to-body channels. While characterizing multiple-antenna body-worn systems, the authors utilized SC to pick the branch with the highest SNR at the receiver. 
\begin{equation}
    R=\max\left( r_{1}, r_{2},\cdots,r_{M} \right)
\end{equation}
Here, $r_M$  represents the signal level measured in the $M^{th}$
 branch of the diversity receiver.
 Our results can be used to evaluate system performance in such scenarios.





    \section{Simulation Results} \label{simulations}
    \begin{figure}[t]
        \centering
        \includegraphics[scale=0.5]{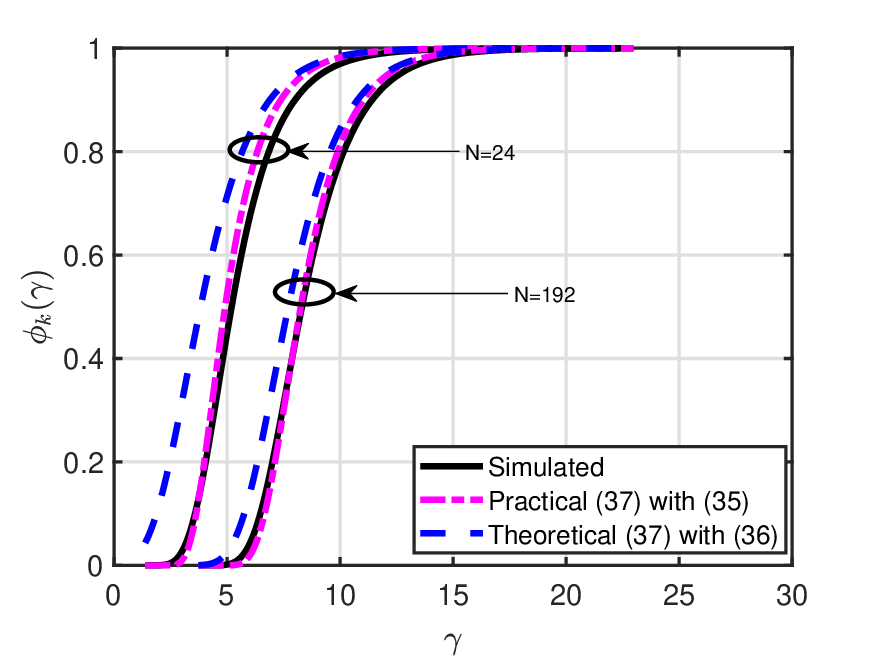}
        \caption{ $ \phi_{k}\left(\gamma\right)$ with $k=1$ using $ \tilde{u}\left ( \gamma \right )$(equations  (\ref{comp_1}) ,(\ref{comp_2}))  for i.n.i.d. RVs }
        \label{fig:inid_comp}
    \end{figure}
     This section presents the simulation results of the asymptotic distribution of $k$-th maximum order statistics of \ac{i.n.i.d.} \ac{SNR} \ac{RVs} in $\kappa-\mu$ fading. Here, we also compare the theoretical \ac{CDF} derived in Section II with the simulated one. For this comparison, we employ $N$ \ac{i.n.i.d.} $\kappa-\mu$ \ac{RVs} characterized by parameters $(\kappa_i, \mu_j)$, where $i$ and $j$ range from 1 to 3.  In Fig. \ref{fig:inid_comp}, we compare the maximum order statistics $(k=1)$ of simulated \ac{CDF} with practical ((\ref{ua3}) with (\ref{comp_1})) and theoretical ((\ref{ua3}) with (\ref{comp_2})) \ac{CDF}s. Note that Eq. (\ref{comp_2}) is a theoretical expression derived for large $N$, whereas Eq. (\ref{comp_1}) is an exact expression that remains valid even for small $N$.  We present the results for \ac{i.n.i.d.}   $\kappa-\mu$ \ac{RVs} with $\kappa_1=0.3$, $\kappa_2=0.7$, $\kappa_3=1.2$ and $\mu_1=0.75$, $\mu_2=1.25$, $\mu_3=1.5$ for $N=24, 192$ with $\bar{\gamma}=2$ dB.\\     
     Note that for all practical purposes, we consider $(\kappa_i, \mu_j)$ come from a finite $N$. Hence, the practical equation curve is much closer to the simulated curve. We can also observe that the difference is very small, and the gap reduces further as we increase $N$.
      \begin{figure}[h]
        \centering
        \includegraphics[scale=0.5]{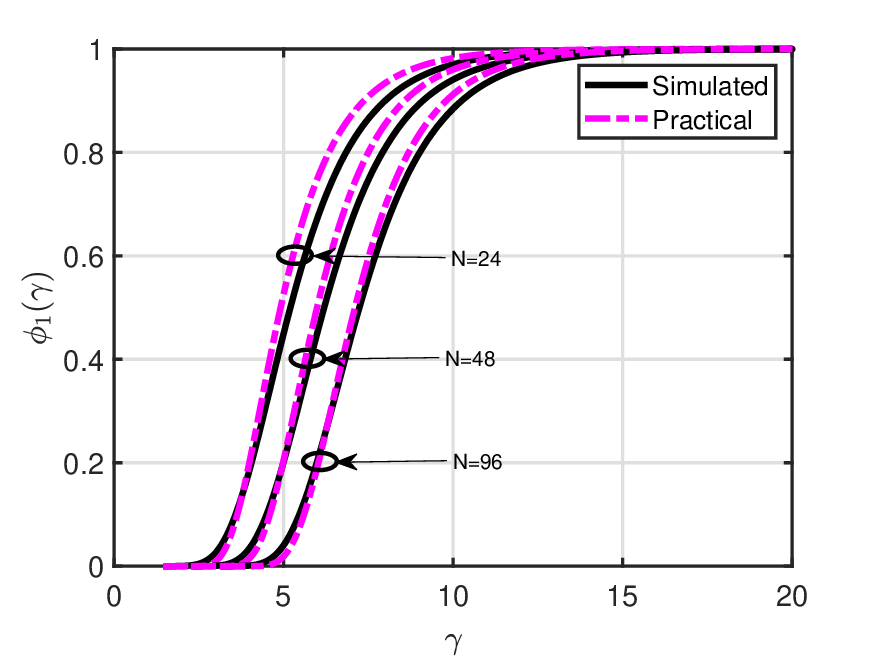}
        \caption{Theoretical and simulated CDFs of the maximum $(k=1)$ order statistics over \ac{i.n.i.d.}$\kappa-\mu$ \ac{RVs} for different N  }
        \label{fig:inid_N}
    \end{figure} 
     \begin{figure}[h]
        \centering
        \includegraphics[scale=0.5]{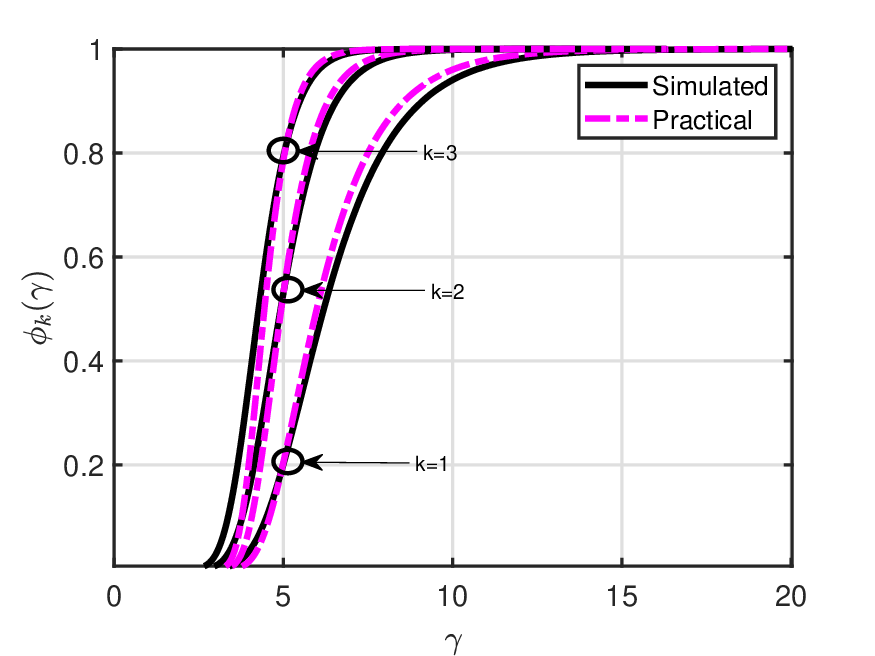}
        \caption{ Theoretical and simulated CDFs of the $k$-th maximum order statistics of \ac{i.n.i.d.} $\kappa-\mu$ \ac{RVs} for different $ k $ }
        \label{fig:inid_k}
    \end{figure}
     \subsection{$k$-th maximum \ac{i.n.i.d.} results}
     Fig. \ref{fig:inid_N} compares the simulated and theoretical \ac{CDF}s of maximum order statistics $(k=1)$ for \ac{i.n.i.d.} $\kappa-\mu$ \ac{RVs} with various values of $N$. Note that as $N$ grows, the  difference between the simulated and theoretical curves  diminishes as expected.   
     Fig. \ref{fig:inid_k} illustrates the comparison between simulated and theoretical CDFs for the $k$-th maximum order statistics of  \ac{i.n.i.d.} $\kappa-\mu$  \ac{RVs} with $N=48$. Results are presented for $k=1,2,3$. Notably, it's observable that the agreement between the simulated curves and the first-order statistics is notably stronger than that of the second and third-order statistics.       
   
    \subsection{Special cases results}
     Here, we present the results for the corollaries derived in section III.
      \subsubsection{$k$-th maximum \ac{i.i.d.} results}
       In Fig. \ref{fig:iid_N}, we compare the simulated and theoretical CDFs of the maximum order statistics $(k=1)$ for \ac{i.i.d.} $\kappa-\mu$ \ac{RVs} across various $N$ values. The simulations were conducted with $\kappa=0.2$ and $\mu=1$, with $\bar{\gamma}=2$ dB.  Again, as $N$ increases, the gap between the simulated and theoretical curves diminishes as expected.

     \begin{figure}[h]
        \centering
        \includegraphics[scale=0.5]{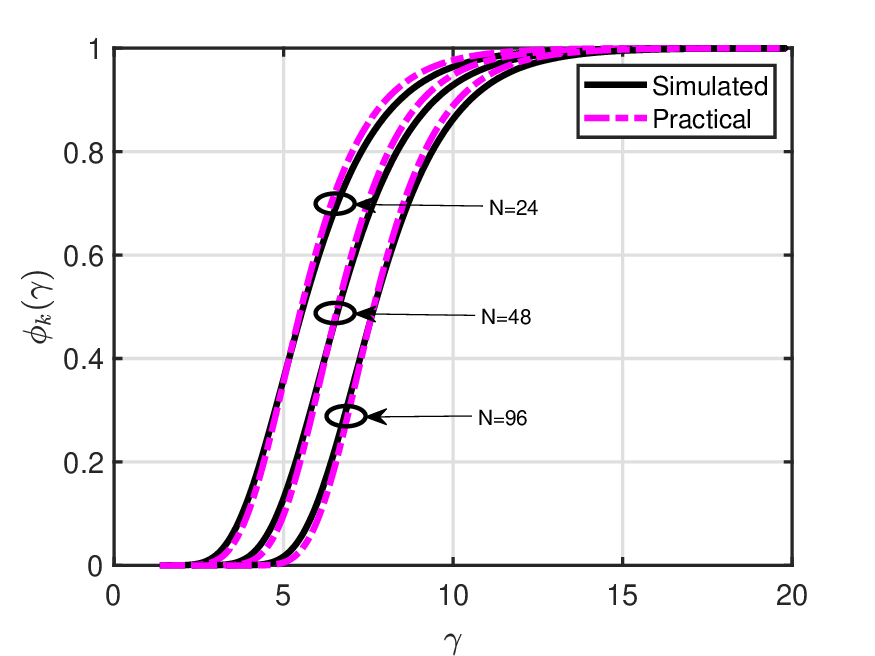}
        \caption{Theoretical and simulated CDFs of the maximum $(k=1)$ order statistics over \ac{i.i.d.} $\kappa-\mu$ \ac{RVs} for different N  }
        \label{fig:iid_N}
    \end{figure}
     Fig. \ref{fig:iid_k} depicts the comparison between simulated and theoretical CDFs for the $k$-th maximum order statistics of \ac{i.i.d.} $\kappa-\mu$ \ac{RVs} with $N=48$. Results are shown for $k=1,2,3$. 
     \begin{figure}[h]
        \centering
        \includegraphics[scale=0.5]{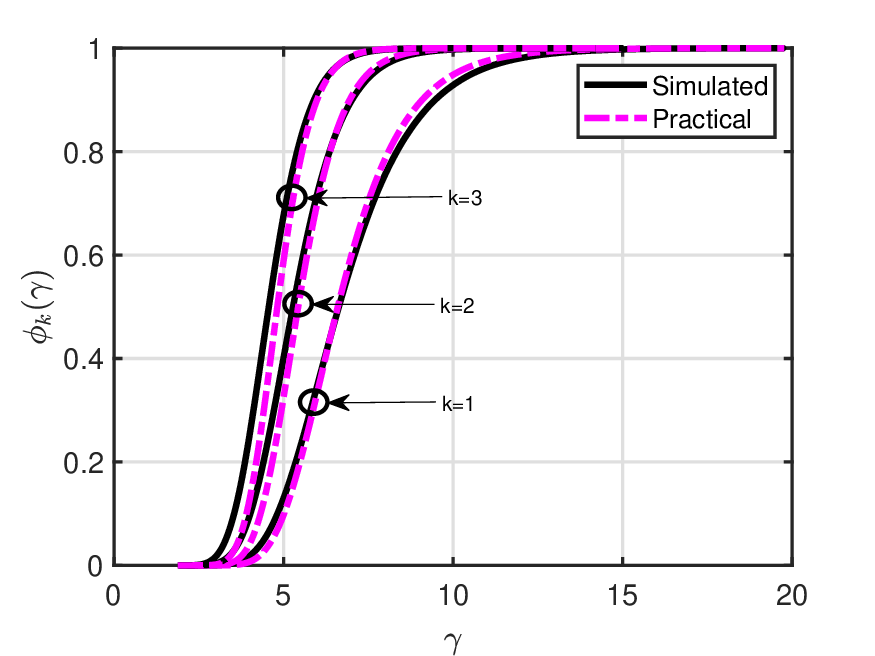}
        \caption{ Theoretical and simulated CDFs of the $k$-th maximum order statistics of \ac{i.i.d.} $\kappa-\mu$ \ac{RVs} for different $ k $ }
        \label{fig:iid_k}
    \end{figure}
    \subsubsection{$\kappa$-\ac{i.n.i.d.} and $\mu$-\ac{i.i.d.}$(\kappa_i,\mu_j=\mu)$}
    \begin{figure}[h]
        \centering
        \includegraphics[scale=0.5]{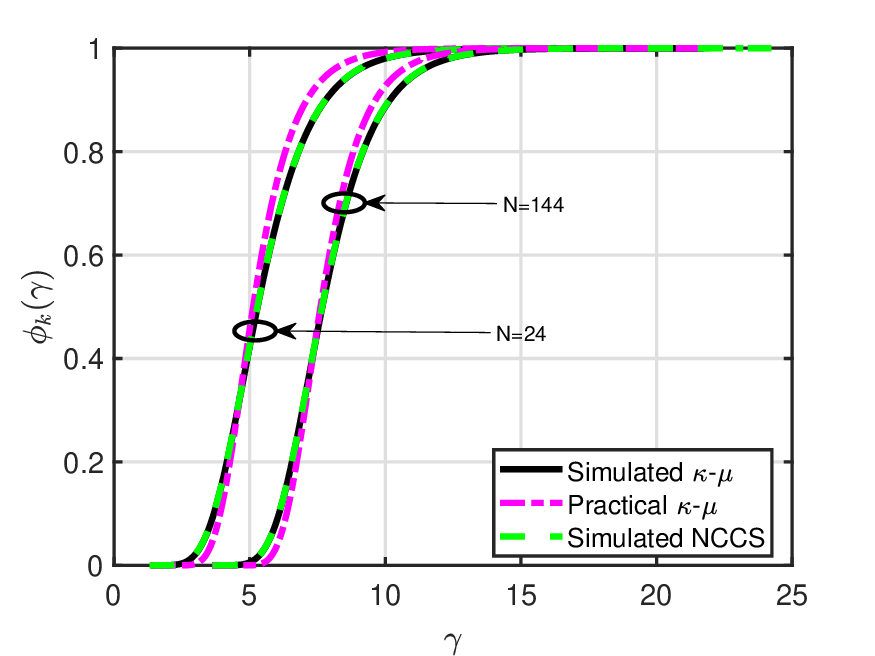}
        \caption{ Comparison of \ac{NCCS} with $\kappa-\mu$ \ac{RVs} }
        \label{fig:nccs_km}
    \end{figure}
     \begin{figure}[h]
        \centering
        \includegraphics[scale=0.5]{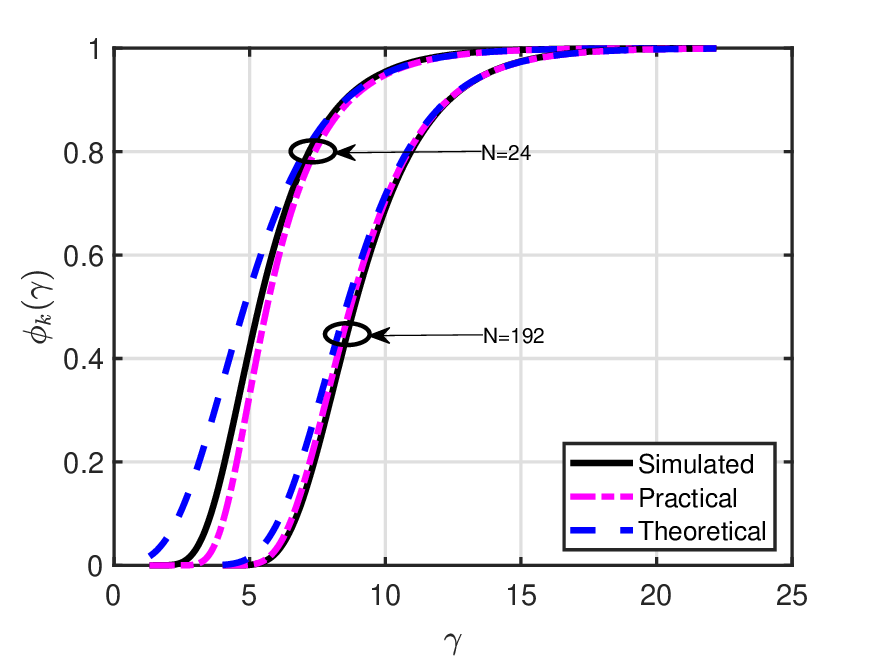}
        \caption{$\kappa$-\ac{i.i.d.} and $\mu$-\ac{i.n.i.d.}  }
        \label{fig:naka}
    \end{figure}
    \ac{NCCS} \ac{RVs} can be expressed in terms squared $\kappa-\mu$ \ac{RVs}.  Fig. \ref{fig:nccs_km} presents the maximum order statistics of \ac{NCCS} \ac{RVs} with squared $\kappa-\mu$ \ac{RVs} along with derived practical expression (corollary \ref{thm_kth_max_2}). Here, we have used $\lambda_1=0.72, \lambda_2=1.28, \lambda_3=2$ as non centrality parameters with \ac{dof}=2 for \ac{NCCS} \ac{RVs}. Accordingly \ac{i.n.i.d.} $\kappa$ values are $\kappa_1=0.36, \kappa_2=0.64, \kappa_3=1$ with \ac{i.i.d.} $\mu=1$ and $\bar{\gamma}=2$ dB. The derived expression is in good agreement with simulations.
    \subsubsection{$\kappa$-\ac{i.i.d.} and $\mu$-\ac{i.n.i.d.}$(\kappa_i=\kappa,\mu_j)$}    
    
    The Nakagami-m distribution can be achieved by setting $\kappa$ to be \ac{i.i.d.} with $\kappa \rightarrow 0$, and $\mu$ to be \ac{i.n.i.d.}. Fig. \ref{fig:naka} illustrates the simulated and theoretical CDFs of maximum order statistics for the specific case of Nakagami-m distribution with $\mu_1=0.8$, $\mu_2=1.5$, and $\mu_3=2$.
    \subsection{Outage Capacity and Average throughput results}
     Here, we present the results for the outage capacity and average throughput expressions derived in (\ref{c_out}), and (\ref{er_2}) for \ac{i.n.i.d.} $\kappa-\mu$ \ac{RVs}. Fig. \ref{fig:outage} presents the results of outage capacity for different values of $k$. 
      \begin{figure}[h]
        \centering
        \includegraphics[scale=0.5]{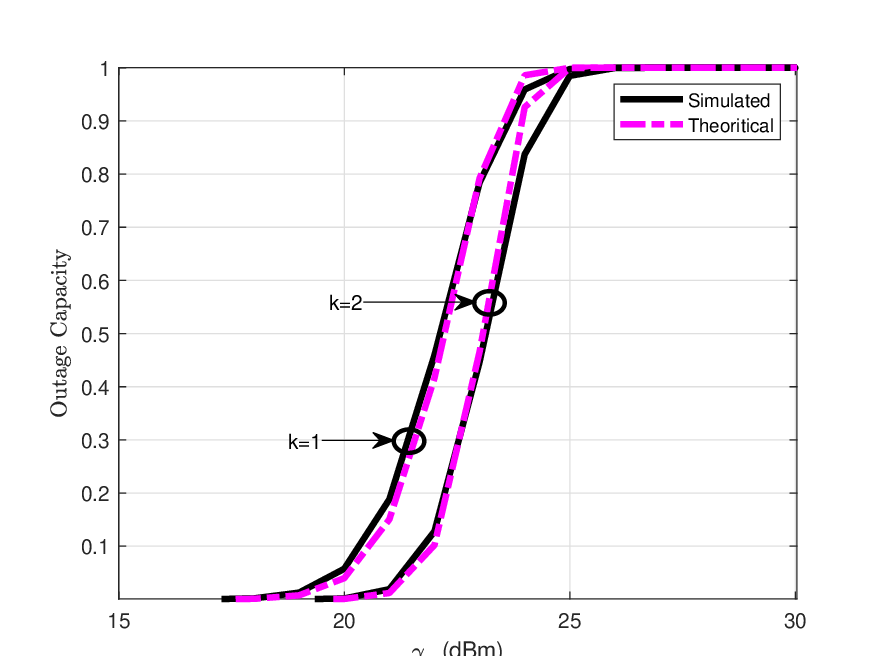}
        \caption{ Outage Capacity with $\gamma_{th}$=0 dB for $\kappa-\mu$ \ac{RVs} for different $ k $ }
        \label{fig:outage}
    \end{figure}
      Next, Fig. \ref{fig:avg_thr_vs_N} presents simulated and theoretical average throughput results of $k$-th maximum order statistics for \ac{i.n.i.d.} $\kappa-\mu$ \ac{RVs} for different $N$ and $k$. Here. we use $\kappa_1=0.8$, $\kappa_2=1.2$, $\kappa_3=2$ and $\mu_1=0.5$, $\mu_2=1$, $\mu_3=1.5$  with $\bar{\gamma}=2$ dB. The results are as expected, and with increasing $N$ the gap between the simulated and theoretical results is reduced. 
    \begin{figure}[h]
        \centering
        \includegraphics[scale=0.5]{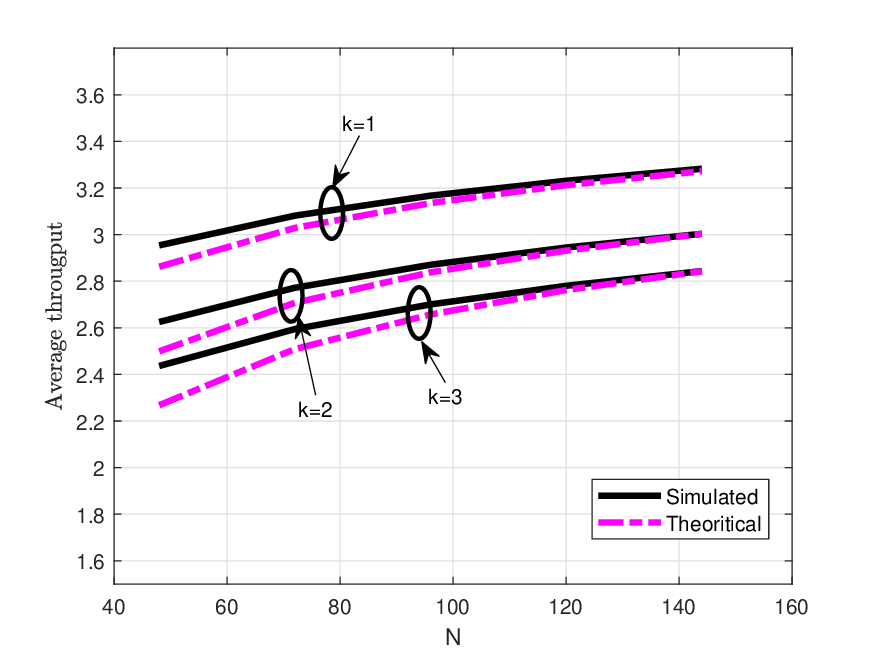}
        \caption{Average throughput for different $k$}
        \label{fig:avg_thr_vs_N}
    \end{figure} 
    \subsection{Gamma applications}
    \subsubsection{Backscatter}
     \begin{figure}[h]
        \centering
        \includegraphics[scale=0.5]{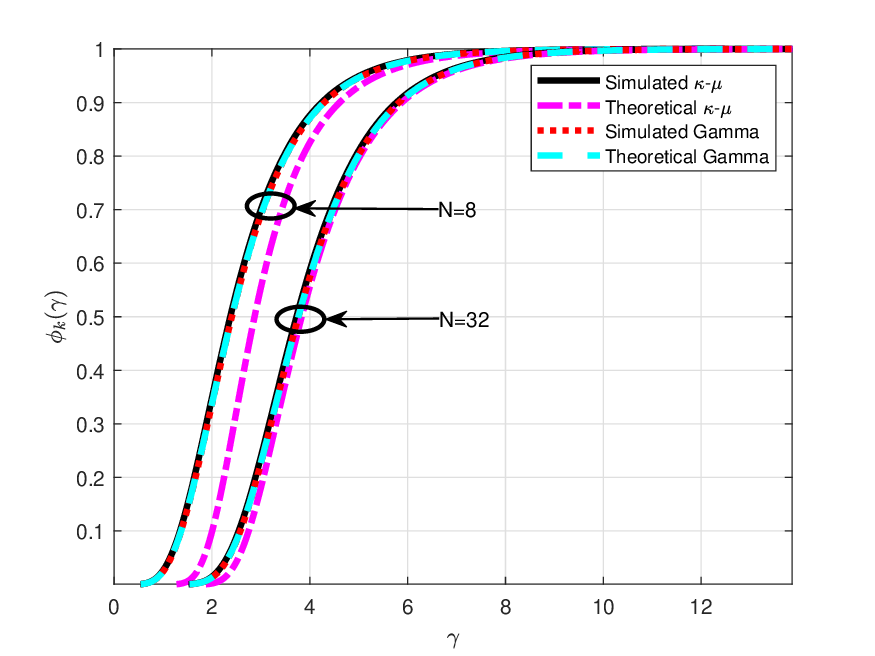}
        \caption{Comparison of $\kappa$–$\mu$ and Gamma Maximum Order Statistics with Application to AmBS \cite{ali2024ergodic}}
        \label{fig:bs}
    \end{figure} 
      Fig. \ref{fig:bs} presents the results for the maximum order statistics of Gamma RVs in the AmBS system considered in \cite{ali2024ergodic}. 
    As illustrated in Fig. \ref{fig:bs}, the simulated results confirm that the maximum order statistics of Gamma RVs coincide with those of the corresponding $\kappa-\mu$ RVs. For the simulation, the parameters $\mu_1=0.8$, $\mu_2=1.2$, and $\mu_3=2$ were used.

    
    \subsubsection{Multi-RIS }
   
     \begin{figure}[h]
        \centering
        \includegraphics[scale=0.5]{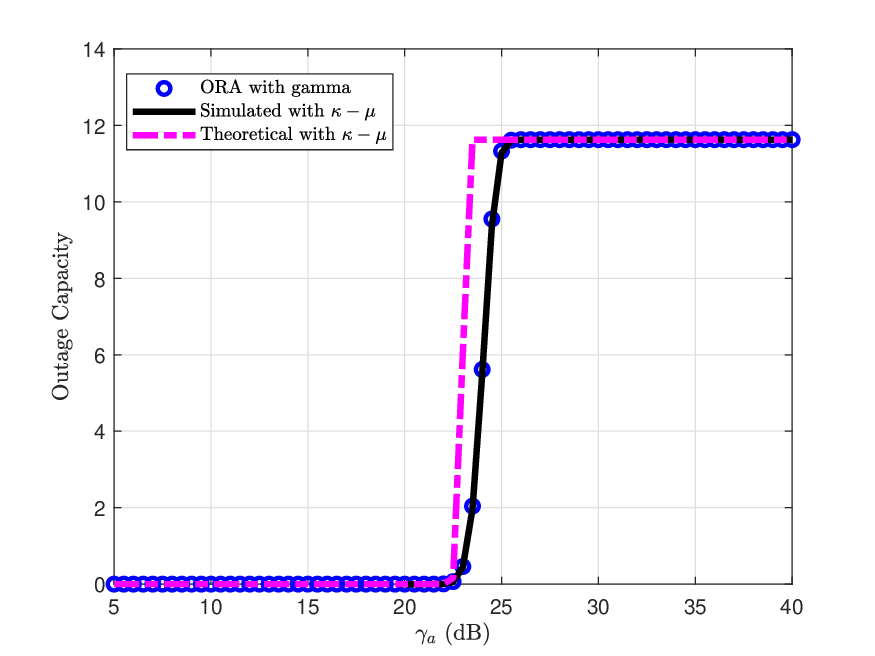}
        \caption{Comparison of $\kappa$–$\mu$ and Gamma Maximum Order Statistics with Application to Multi-RIS \cite{do2021multi}}
        \label{fig:ORA}
    \end{figure} 
     Fig. \ref{fig:ORA} illustrates the outage capacity performance of the ORA scheme presented in \cite{do2021multi} together with our simulation results. The authors in \cite{do2021multi}derived an approximate closed-form expression for the ergodic capacity of the ORA scheme, considering $V_n$ as gamma random variables. We compared their results with our findings, where the gamma distribution is treated as a special case of the $\kappa-\mu$ distribution, as illustrated in Fig. \ref{fig:ORA}. Once again, the simulations confirm that our derived expressions closely match the results reported in \cite{do2021multi}.
    
    \section{{Conclusions and Future works}}\label{conclusion}
    This paper derived the asymptotic distribution of the $k$-th maximum order statistics of i.n.i.d. \ac{SNR} \ac{RVs} in $\kappa-\mu$ fading using EVT. The presented results are applicable in various scenarios where determining order statistics in a $\kappa-\mu$ fading environment is necessary. Furthermore, we also derived performance metrics, such as outage probability and average throughput. The practical value of our findings was shown across diverse applications, including  MIMO systems, relay networks, body-to-body communications, UAV-supported relays, backscatter systems, and multi-RISs. 
    Moving forward, we aim to expand our investigation of \ac{i.n.i.d.} order statistics to encompass more intricate fading scenarios, such as $\kappa-\mu$ shadow fading.

    \begin{appendices}
    \section{Proof for Theorem \ref{thm_kth_max}} \label{proof_main}
    Using CDF expression in  (\ref{eq1}) and applying in (\ref{ua2}),
    \begin{equation}\label{eq3}
        \tilde{u}(\gamma)=\lim _{N\rightarrow \infty} \sum_{n=1}^N Q_{\mu_n}\left(\sqrt{2\kappa_n\mu_n}, \sqrt{\frac{2(1+\kappa_n)\mu_n (a_N\gamma+b_N)}{\bar{\gamma}_n}}\right).
    \end{equation}
    Let $\omega_n=\sqrt{2\kappa_n\mu_n}$ and $\zeta_n=\sqrt{\frac{2(1+\kappa_n)\mu_n (a_N\gamma+b_N)}{\bar{\gamma}_n}}$. 
    In addition, based on Mejzler's theorem \cite[Chapter 5]{de2006extreme}, which specifies the normalizing constants $a_N$ and $b_N$,  we consider the asymptotic condition $b_N \to \infty$ as $N \to \infty$. Therefore, the second argument of the Marcum-$Q$ function also becomes unbounded with increasing $N$. In this limit, the Marcum-$Q$ function $\mathcal{Q}_{\mu_n}(\omega_{n}, \zeta_{n})$ can be represented using the $Q$ function \cite[(A 27)]{nuttall} for $\zeta_{n} \to \infty$,    
    \begin{align}\label{marc}
         Q_{\mu_n}(\omega_{n}, \zeta_{n})& =\left(\frac{\zeta_{n}}{\omega_{n}}\right)^{\mu_n-\frac{1}{2}} Q(\zeta_{n}\!-\!\omega_{n}).
    \end{align}
   Next, by expressing the $Q(\cdot)$ function in terms of the complementary error function $\operatorname{erfc}(\cdot)$ \cite[Eq. 12.26]{molisch2012wireless}, 
    $ \  Q(\zeta_{n}\!-\!\omega_{n}) = \frac{1}{2} \operatorname{erfc}\left(\frac{\zeta_{n}\!-\!\omega_{n}}{\sqrt{2}}\right)\,$, 
    and utilizing the asymptotic expansion of $\operatorname{erfc}(\cdot)$ given in \cite{erfc}
\begin{equation}
\operatorname{erfc}(z)=\frac{1}{\sqrt{\pi}z}\exp(-z^{2})
\left(1+\mathcal{O}\left(\frac{1}{z^{2}}\right)\right),
\qquad |z|\to\infty,
\end{equation}
the asymptotic expression for the  $Q$ function can be written as
\begin{equation}
Q(z)= \frac{1}{\sqrt{2\pi}z}\exp\left(-\frac{z^{2}}{2}\right).
\end{equation}


Substituting $Q(z)$ in (\ref{marc})
    \begin{align}\label{marc2}
         Q_{\mu_n}(\omega_{n}, \zeta_{n})&={\color{orange}\left(\frac{\zeta_{n}}{\omega_{n}}\right)^{\mu_n-\frac{1}{2}}}{\color{blue}\frac{1}{\sqrt{2\pi }\left ( \zeta_{n}\!-\!\omega_{n} \right )}}{\color{red}\exp{\left(-\frac{(\zeta_{n}\!-\!\omega_{n})^2}{2}\right)}}.
    \end{align}
    Further, substituting $\omega_n$, $\zeta_n$ and using the fact that $a_{N} \gamma+b_{N}>> 2\kappa_n \mu_n$  as $b_N \to \infty$ for all choices of $a_{N}$ and $b_{N}$, the terms in Eq. (\ref{marc2}) can be written as 
   \begin{align}\label{qfun}
        {\color{orange}\left(\frac{\zeta_{n}}{\omega_{n}}\right)^{\mu_n-\frac{1}{2}}}&=\left ( \frac{\sqrt{2\left ( 1+\kappa_n \right )\mu_{n}(a_N \gamma+b_N)}}{\sqrt{\bar{\gamma}_n}\sqrt{2\kappa_n \mu_n}} \right )^{\mu_{n}-\frac{1}{2}},
    \end{align} 
    \begin{align}\label{qfun}
       {\color{blue}\frac{1}{\sqrt{2\pi }\left ( \zeta_{n}-\omega_{n} \right )}}&= \frac{1}{\sqrt{2\pi}\left ( \frac{\sqrt{2\left ( 1+\kappa_n \right )\mu_{n}(a_N \gamma+b_N)}}{\sqrt{\bar{\gamma}_n}}-\sqrt{2\kappa_n \mu_n} \right )},
    \end{align} 
    and
    \begin{align}\label{qfun}
        &{\color{orange}\left(\frac{\zeta_{n}}{\omega_{n}}\right)^{\mu_n-\frac{1}{2}}}{\color{blue}\frac{1}{\sqrt{2\pi }\left ( \zeta_{n}\!-\!\omega_{n} \right )}}\nonumber \\&=\frac{\sqrt{1+\kappa_n}^{\mu_n-\frac{3}{2}}}{\sqrt{\kappa_n}^{\mu_n-\frac{1}{2}}}\frac{1}{2\sqrt{\pi \mu_n}} \sqrt{\frac{a_N \gamma+ b_N}{\bar{\gamma}_n}}^{\mu_n-\frac{3}{2}} .
    \end{align} 
    Also, using the fact that $\left(a_{N} \gamma+b_{N}\right)\approx b_{N}$ as $b_{N}>>a_{N} $ we obtain
    \begin{align}\label{qfun2}        
         &{\color{red}\exp \left(-\frac{\left(\zeta_{n}\!-\!\omega_{n}\right)^2}{2}\right)} =\exp \left(-\frac{\zeta_{n}^2}{2}\right)  \exp \left(-\frac{\omega_{n}^2}{2}\right)  \exp \left(\omega_{n} \zeta_{n}\right)\nonumber \\
         &=\exp \left(-\frac{2(1+\kappa_n)\mu_n\left(a_{N} \gamma+b_{N}\right)}{2\bar{\gamma}_n }\right) \exp \left(-\frac{2\kappa_{n}\mu_n}{2 }\right) \nonumber \\& \quad \quad\exp \left(\sqrt{\frac{2(1+\kappa_n)\mu_n}{\bar{\gamma}_n}}  \sqrt{a_{N} \gamma+b_{N}} \sqrt{2\kappa_n\mu_n}\right)\nonumber \\
         &\approx\exp \left(-\frac{(1+\kappa_n)\mu_n}{\bar{\gamma}_n}\left(a_{N} \gamma+b_{N}\right)\right) \exp \left(-\kappa_{n}\mu_n\right) \nonumber \\ &\quad \quad\exp \left(\sqrt{\frac{2(1+\kappa_n)\mu_n}{\bar{\gamma}_n}}  \sqrt{b_{N}} \sqrt{2\kappa_n\mu_n}\right).       
    \end{align}        
      Now, substituting (\ref{marc2}) in (\ref{eq3}) using (\ref{qfun}), (\ref{qfun2}), we can rewrite (\ref{eq3}) as 
    \begin{align}\label{utilgam}
        \tilde{u}(\gamma)=&\lim_{N \rightarrow \infty} \sum_{n=1}^N \frac{\sqrt{1+\kappa_n}^{\mu_n-\frac{3}{2}}}{\sqrt{\kappa_n}^{\mu_n-\frac{1}{2}}}\frac{1}{2\sqrt{\pi \mu_n}} \sqrt{\frac{a_N \gamma+ b_N}{\bar{\gamma}_n}}^{\mu_n-\frac{3}{2}}\nonumber \\ &\exp \left(-\frac{(1+\kappa_n)\mu_n}{\bar{\gamma}_n}\left(a_{N} \gamma+b_{N}\right)\right)  \exp \left(-\kappa_{n}\mu_n\right)\nonumber \\ & \quad \quad  \exp \left(\sqrt{\frac{2(1+\kappa_n)\mu_n}{\bar{\gamma}_n}}  \sqrt{b_{N}} \sqrt{2\kappa_n\mu_n}\right).
    \end{align}        
    Similar to \cite{peppas2011sum}, we assume average SNR $\bar{\gamma}_n=\bar{\gamma}$ for all $n$.   Recall  $N_{ij}$ represents the number of times  pair $\left(\kappa_i, \mu_j\right)$ occurs among  $N$ values i.e.
    $$N_{ij}=\sum_{n=1}^N \mathbb{I}_{\kappa_{n}=\kappa_i, \mu_{n}=\mu_j} \quad  1 \leq i \leq P , 1 \leq j \leq L.$$
    The single summation in  (\ref{utilgam}), involving the parameter pairs $(\kappa_n,\mu_n)$, can be reformulated as a double summation encompassing all possible combinations of the parameters $(\kappa_i,\mu_j)$
    \begin{align}
        & \tilde{u}(\gamma)=\lim _{N \rightarrow \infty} \sum_{i=1}^{P } \sum_{j=1}^{L } N_{ij} \frac{\sqrt{1+\kappa_i}^{\mu_j-\frac{3}{2}}}{\sqrt{\kappa_i}^{\mu_j-\frac{1}{2}}}\frac{1}{2\sqrt{\pi \mu_j}}  \exp \left(-\kappa_{i}\mu_j\right) \nonumber \\ & \exp\left ( \log \left ( \left ( \frac{a_N \gamma+b_N}{\bar{\gamma}}\right )^{\frac{2\mu_{j} -3}{4}} \right ) \right ) 
        \exp \left(-\frac{(1+\kappa_i)\mu_j}{\bar{\gamma}}\left(a_{N} \gamma \right)\right)\nonumber\\& \exp \left(-\frac{(1+\kappa_i)\mu_j}{\bar{\gamma}}\left(b_{N}\right)\right)  \exp \left(\sqrt{\frac{2(1+\kappa_i)\mu_j}{\bar{\gamma}}}  \sqrt{b_{N}} \sqrt{2\kappa_i\mu_j}\right).
     \end{align}
    Suppose, we choose the normalizing constants $a_{N}$, $b_{N}$ in the following form        
    \begin{align}\label{andb}
        a_{N} &=\frac{\bar{\gamma}}{(1+\tilde{\kappa})\tilde{\mu}} \nonumber \\
        b_{N} &=a_N c_N\nonumber \\ &=a_N \left(\log (\tilde{N})-c_0  \log (\log (\tilde{N}))+c_1 \sqrt{\log (\tilde{N})}-c_2\right),
    \end{align}
    
    with constants $ c_0, c_1$, and $c_2$.
    Let  $\tilde{\kappa}$, $\tilde{\mu}$ be the smallest among $\left\{\kappa_{1},\cdots,\kappa_{P}\right\}$ and $\left\{\mu_{1},\cdots,\mu_{L}\right\}$ respectively. Recall from Theorem \ref{thm_kth_max} statement, $\tilde{N}$ represents the number of times the pair $(\kappa_n=\tilde{\kappa},\mu_n=\tilde{\mu})$ occurs among $N$, note we assume $\tilde{N}\to \infty $ as $N \to \infty $.  For the above choice of normalizing constants, let us 
    evaluate $ \tilde{u}(\gamma)= \sum_{(i,j) \rightarrow (1,1)}^{(P \times L) }\tilde{u}(\gamma)_{ij}$, where $\tilde{u}(\gamma)_{ij}$ corresponds to the pair ($\kappa_i,\mu_j$).
    \begin{align}
        &\tilde{u}(\gamma)_{ij}=\lim _{N \rightarrow \infty}  N_{ij} \frac{\sqrt{1+\kappa_i}^{\mu_j-\frac{3}{2}}}{\sqrt{\kappa_i}^{\mu_j-\frac{1}{2}}}\frac{1}{2\sqrt{\pi \mu_j}}  \exp \left(-\kappa_{i}\mu_j\right)\nonumber \\& \exp\left ( \frac{2\mu_{j} -3}{4}\log \left ( \frac{a_N \gamma+b_N}{\bar{\gamma}} \right ) \right )
        \exp \left(-\frac{(1+\kappa_i)\mu_j}{\bar{\gamma}}\left(a_{N} \gamma \right)\right) \nonumber \\&\exp \left(-\frac{(1+\kappa_i)\mu_j}{\bar{\gamma}}\left(b_{N}\right)\right)  \exp \left(\sqrt{\frac{2(1+\kappa_i)\mu_j}{\bar{\gamma}}}  \sqrt{b_{N}} \sqrt{2\kappa_i\mu_j}\right).
    \end{align}
    
    Again, since $b_{N}>>a_{N}$, we approximate $a_{N} \gamma+b_{N}$ as $b_{N}$. Subsequently, substituting the expression for $a_{N}$ from (\ref{andb}), we obtain
    \begin{align}
        &\tilde{u}(\gamma)_{ij}=\lim _{N \rightarrow \infty}  N_{ij} \frac{\sqrt{1+\kappa_i}^{\mu_j-\frac{3}{2}}}{\sqrt{\kappa_i}^{\mu_j-\frac{1}{2}}}\frac{1}{2\sqrt{\pi \mu_j}}  \exp \left(-\kappa_{i}\mu_j\right)\nonumber \\& \exp\left ( \frac{2\mu_{j} -3}{4}\log \left ( \frac{b_N}{\bar{\gamma}} \right ) \right ) \exp\left (-\frac{\left ( 1+\kappa_i \right )\mu_j}{\left ( 1+\tilde{\kappa} \right )\tilde{\mu}}\gamma\right )\nonumber \\&\exp\left (-\frac{\left ( 1+\kappa_i \right )\mu_j}{\left ( 1+\tilde{\kappa} \right )\tilde{\mu}} c_N\right )\exp \left(\sqrt{\frac{2(1+\kappa_i)\mu_j}{\bar{\gamma}}}  \sqrt{b_{N}} \sqrt{2\kappa_i\mu_j}\right).
    \end{align}
    Let $\theta_{ij}=\frac{\left ( 1+\kappa_i \right )\mu_j}{\left ( 1+\tilde{\kappa} \right )\tilde{\mu}}$, by definitions of $\tilde{\kappa}$ and $\tilde{\mu}$, $\theta_{ij}\ge 1$. Then 
    \begin{align}
        &\tilde{u}(\gamma)_{ij}=\lim _{N \rightarrow \infty}  N_{ij} \frac{\sqrt{1+\kappa_i}^{\mu_j-\frac{3}{2}}}{\sqrt{\kappa_i}^{\mu_j-\frac{1}{2}}}\frac{1}{2\sqrt{\pi \mu_j}}  \exp \left(-\kappa_{i}\mu_j\right)\nonumber \\& \exp\left ( \frac{2\mu_{j} -3}{4}\log \left ( \frac{b_N}{\bar{\gamma}} \right ) \right ) \exp\left (-\theta_{ij}\gamma\right )\exp\left (-\theta_{ij} c_N\right )\nonumber \\&\exp \left(  2 \sqrt{\theta_{ij}\kappa_i\mu_j}\sqrt{c_{N}}\right).
    \end{align}   
   
    Substituting constant $c_{N}$ in $\exp\left (-\theta_{ij} c_N\right )$,
     \begin{align}\label{ku}
        &\tilde{u}(\gamma)_{ij}=\lim _{N \rightarrow \infty}  N_{ij} \frac{\sqrt{1+\kappa_i}^{\mu_j-\frac{3}{2}}}{\sqrt{\kappa_i}^{\mu_j-\frac{1}{2}}}\frac{1}{2\sqrt{\pi \mu_j}}  \exp \left(-\kappa_{i}\mu_j\right)\nonumber \\&\exp\left (-\theta_{ij}\gamma\right )\exp \left(  2 \sqrt{\theta_{ij}\kappa_i\mu_j}\sqrt{c_{N}}\right) \exp\left ( \frac{2\mu_{j} -3}{4}\log \left ( \frac{b_N}{\bar{\gamma}} \right ) \right )\nonumber \\& \exp\left (-\theta_{ij} \left(\log (\tilde{N})-c_0  \log (\log (\tilde{N}))+c_1 \sqrt{\log (\tilde{N})}-c_2\right)\right ).
    \end{align}
    \begin{figure*}[t]    
     \begin{align}\label{eq39}
        \tilde{u}(\gamma)_{ij}=&\lim _{N \rightarrow \infty} \exp\left (-\theta_{ij}\gamma\right )  N_{ij} \frac{\sqrt{1+\kappa_i}^{\mu_j-\frac{3}{2}}}{\sqrt{\kappa_i}^{\mu_j-\frac{1}{2}}}\frac{1}{2\sqrt{\pi \mu_j}}  \exp \left(-\kappa_{i}\mu_j\right)\nonumber  \exp\left [ \log  \left (\frac{b_N}{\bar{\gamma}}\right )^{\frac{2\mu_{j} -3}{4}}+ \log  \left (\log (\tilde{N})\right )^{\theta_{ij}c_0}  \right ]  \exp\left [  \log  \left (\frac{1}{\tilde{N}^{\theta_{ij}}}\right )  \right ]\\ &  \exp\left (- \theta_{ij}c_1 \sqrt{\log \left ( \tilde{N} \right )} \right ) \exp \left ( \theta_{ij}c_2 \right )
        \exp \left(  2 \sqrt{\theta_{ij}\kappa_i\mu_j}\sqrt{c_{N}}\right).
    \end{align}
    \end{figure*}
    Rearranging the terms in (\ref{ku}) yields (\ref{eq39}). Further simplification of (\ref{eq39}) leads to (\ref{mn}). Note that, to prove Theorem \ref{thm_kth_max}, we need to show (\ref{mn}) evaluates to a finite value for all $\tilde{u}(\gamma)_{ij}$. It is apparent that all terms other than Term 1,2, and 3 in (\ref{mn}) are finite for all $\tilde{u}(\gamma)_{ij}$ for all $N$. Hence, to prove Theorem \ref{thm_kth_max}, we need to show Terms 1,2, and 3 are finite as $N\rightarrow \infty$.
     \begin{figure*}[t]
     \begin{align}\label{mn}
        \tilde{u}(\gamma)_{ij}=&\lim _{N \rightarrow \infty} \exp\left (-\theta_{ij}\gamma\right )  \frac{\sqrt{1+\kappa_i}^{\mu_j-\frac{3}{2}}}{\sqrt{\kappa_i}^{\mu_j-\frac{1}{2}}}\frac{1}{2\sqrt{\pi \mu_j}}  \exp \left(-\kappa_{i}\mu_j\right) \exp \left ( \theta_{ij}c_2 \right )  {\color{orange}\underbrace{\frac{N_{ij}}{(\tilde{N})^{\theta_{ij}}}}_{\text{Term-1}}}  {\color{blue}\underbrace{\exp\left [ \log  \left (\frac{b_N}{\bar{\gamma}}\right )^{\frac{2\mu_{j} -3}{4}}+ \log  \left (\log (\tilde{N})\right )^{\theta_{ij}c_0}  \right ] }_{\text{Term-2}}} \nonumber \\& {\color{red}\underbrace{\exp\left (- \theta_{ij}c_1 \sqrt{\log \left ( \tilde{N} \right )}+ 2 \sqrt{\theta_{ij}\kappa_i\mu_j}\sqrt{c_{N}} \right )}_{\text{Term-3}}} .        
    \end{align} 
    \hrulefill
    \end{figure*}
    Similar to \cite{subhash2022asymptotic}, and from Theorem \ref{thm_kth_max},  $\tilde{N}\to \infty $ as $N \to \infty $. 
    We simplify Terms 1, 2, and 3 in  (\ref{mn}).  Recall from Theorem \ref{thm_kth_max} statement, $\tilde{\theta}=\min \left \{ \theta_{ij} \right \}_{(i,j)=(1,1)}^{(P, L)}$, and note that $\theta_{ij}$ can only take values greater than or equal to one.  Note that in Term-1, $\lim _{N \rightarrow \infty} \frac{N_{ij}}{(\tilde{N})^{\theta_{ij}}} \rightarrow 1$ when $\theta_{ij}=1$ (since $N_{ij}=\tilde{N}$ whenever $\theta_{ij}=1$)  and $\lim _{N \rightarrow \infty} \frac{N_{ij}}{(\tilde{N})^{\theta_{ij}}} \rightarrow 0$ when $\theta_{ij}\neq 1$.\\
    Let $c_{0}=\frac{-(2\tilde{\mu}-3)}{4}$ and simplify the Term-2. Whenever $\mu_j=\tilde{\mu}$, Term-2  can be written as 
    \begin{align}\label{ne}
        &\exp\left [ \log  \left (\frac{b_N}{\bar{\gamma}}\right )^{\frac{2\mu_{j} -3}{4}}+ \log  \left (\log (\tilde{N})\right )^{\theta_{ij}c_0}  \right ]\nonumber \\ &=\exp \left ( \frac{2 \tilde{\mu} -3}{4} \log \left ( \frac{b_N}{\bar{\gamma}(\log (\tilde{N}))^{\theta_{ij} }} \right ) \right ).
    \end{align}
    substitute $b_{N}$ in Eq. (\ref{ne}), we get Eq. (\ref{eq:s}).    
        \begin{equation}\label{eq:s}
             \lim _{N \rightarrow \infty} \exp \left[{\frac{2\tilde{\mu}-3}{4}}\log \left(\frac{a_N}{\bar{\gamma}}S\right)\right],
        \end{equation}
     where $S=\frac{\log (\tilde{N})- c_0\log (\log (\tilde{N}))+ c_{1} \sqrt{\log (\tilde{N})}- c_2}{(\log (\tilde{N}))^{\theta_{ij} }}$.
    The limit in (\ref{eq:s}) needs to be evaluated for different values of $\theta_{ij}$.  When $(\tilde{\kappa}=\kappa_i, \tilde{\mu}=\mu_j)$, then $\theta_{ij}=1$ and $\lim_{N\rightarrow \infty } S\rightarrow 1$. Hence, $\lim_{N\rightarrow \infty }\exp \left ( \log \left ( \frac{a_N}{\bar{\gamma}} S \right )^{\frac{2\tilde{\mu}-3}{4}} \right )\rightarrow (\frac{a_N}{\bar{\gamma}})^{\frac{2\tilde{\mu}-3}{4}}$, which is finite. For all $(\tilde{\kappa}\neq\kappa_i, \tilde{\mu}\neq\mu_j)$, then $\theta_{ij}>1$ and $\lim_{N\rightarrow \infty } S\rightarrow 0$.
    \par Term-3 can be written as 
    \begin{equation*}
        \exp\left [ \theta_{ij}\left ( 2\sqrt{\frac{\kappa_i\mu_j}{\theta_{ij}}}\sqrt{c_N}-c_1\sqrt{\log\left ( \tilde{N} \right )} \right ) \right ].
    \end{equation*}
    Let $c_1=2 \sqrt{\frac{\tilde{\kappa}\tilde{\mu}}{\tilde{\theta}}}$ and simplifying Term-3 gives Eq. (\ref{ts}).
    \begin{figure*} [b]
    \hrulefill
        \begin{multline}\label{ts}
             \lim _{N \rightarrow \infty} \exp\left [ \theta_{ij}\left ( 2\sqrt{\frac{\kappa_i\mu_j}{\theta_{ij}}}\sqrt{c_N}-2 \sqrt{\frac{\tilde{\kappa}\tilde{\mu}}{\tilde{\theta}}}\sqrt{\log\left ( \tilde{N} \right )} \right ) \right ]  \\= \lim _{N \rightarrow \infty} \exp \left[2\theta_{ij} \left(\sqrt{\frac{\kappa_i\mu_j}{\theta_{ij}}} \sqrt{\left ( \sqrt{\log (\tilde{N})}+\frac{c_1}{2} \right )^{2}-\frac{c_1^2}{4}-c_0 \log (\log(\tilde{N}))-c_{2}}-\sqrt{\frac{\tilde{\kappa}\tilde{\mu}}{\tilde{\theta}}} \sqrt{\log (\tilde{N})}\right)\right]
        \end{multline}
    \end{figure*}
    Let $x_1=\sqrt{\log (\tilde{N})}$ and $x_2=\frac{c_1}{2}$, then
    \begin{figure*}[b]   
        \begin{multline}\label{eq100}
            \lim _{N \rightarrow \infty} \exp \left[2\theta_{ij} \left(\sqrt{\frac{\kappa_i\mu_j}{\theta_{ij}}} \sqrt{(x_1+x_2)^2-x_2^2-c_0 \log(x_1^2)-c_2}-\sqrt{\frac{\tilde{\kappa}\tilde{\mu}}{\tilde{\theta}}} x_1\right)\right]\\
            = \lim _{N \rightarrow \infty} \exp \left[2\theta_{ij} \left(\sqrt{\frac{\kappa_i\mu_j}{\theta_{ij}}}x_1 \left(1+\frac{2 x_2}{x_1}-{\color{purple}\frac{c_0\log x_1^2}{x_1^2}-\frac{c_2}{x_1^2}}\right)^{\frac{1}{2}}-\sqrt{\frac{\tilde{\kappa}\tilde{\mu}}{\tilde{\theta}}} x_1\right)\right]
        \end{multline}
     \end{figure*}
    Note that $x_1 \rightarrow \infty$ as $N \rightarrow \infty $, hence ignoring the purple color terms in  (\ref{eq100}), we get 
   \begin{align}\label{eq45}
        \lim _{N \rightarrow \infty} \exp \left[2\theta_{ij} \left(\sqrt{\frac{\kappa_i\mu_j}{\theta_{ij}}}x_1 {\color{green}\left(1+\frac{2 x_2}{x_1}\right)^{\frac{1}{2}}}-\sqrt{\frac{\tilde{\kappa}\tilde{\mu}}{\tilde{\theta}}} x_1\right)\right].
    \end{align}
     Next, applying the binomial expansion to the green-colored term, ${\color{green}\left(1+\frac{2 x_2}{x_1}\right)^{\frac{1}{2}}}=\left(1+\frac{1}{2} \cdot \frac{2 x_2}{x_1}+\cdots\right)=1+\frac{x_2}{x_1}$, and using the fact that $x_1 \rightarrow \infty$ as $N \rightarrow \infty $, the higher-order terms become negligible. Therefore, only the first two terms of the expansion are retained. Substituting this in (\ref{eq45}), we get
   \begin{multline}
      \lim _{N \rightarrow \infty} \exp \left[2\theta_{ij} \left(\sqrt{\frac{\kappa_i\mu_j}{\theta_{ij}}} (x_1+x_2)-\sqrt{\frac{\tilde{\kappa}\tilde{\mu}}{\tilde{\theta}}} x_1\right)\right].
    \end{multline}    
    After substituting $x_1$ and $x_2$ values, Term-3 simplifies to 
    \begin{equation*}
        \exp \left[2\theta_{ij} \left(\sqrt{\frac{\kappa_i\mu_j}{\theta_{ij}}} -\sqrt{\frac{\tilde{\kappa}\tilde{\mu}}{\tilde{\theta}}} \right)\sqrt{\log (\tilde{N})}+2\theta_{ij}\sqrt{\frac{\kappa_i\mu_j}{\theta_{ij}}}\sqrt{\frac{\tilde{\kappa}\tilde{\mu}}{\tilde{\theta}}}\right].
    \end{equation*}
    Finally, substituting the Term-2 and 3 in (\ref{mn}), we have  
    \begin{align}\label{mn_2}
       & \tilde{u}(\gamma)_{ij}=\lim _{N \rightarrow \infty} \exp\left (-\theta_{ij}\gamma\right )  \frac{\sqrt{1+\kappa_i}^{\mu_j-\frac{3}{2}}}{\sqrt{\kappa_i}^{\mu_j-\frac{1}{2}}}\frac{1}{2\sqrt{\pi \mu_j}}  \exp \left(-\kappa_{i}\mu_j\right)\nonumber \\& \exp \left ( \theta_{ij}c_2 \right )  {\color{orange}\underbrace{\frac{N_{ij}}{(\tilde{N})^{\theta_{ij}}}}_{\text{Term-1}}}  {\color{blue}\underbrace{\exp \left ( \log \left ( \frac{a_N}{\bar{\gamma}} S \right )^{\frac{2\tilde{\mu}-3}{4}} \right ) }_{\text{Term-2}}}  \nonumber \\&{\color{red}\underbrace{\exp \left[2\theta_{ij} \left(\sqrt{\frac{\kappa_i\mu_j}{\theta_{ij}}} -\sqrt{\frac{\tilde{\kappa}\tilde{\mu}}{\tilde{\theta}}} \right)\sqrt{\log (\tilde{N})}+2\theta_{ij}\sqrt{\frac{\kappa_i\mu_j}{\theta_{ij}}}\sqrt{\frac{\tilde{\kappa}\tilde{\mu}}{\tilde{\theta}}}\right]}_{\text{Term-3}}} .        
    \end{align}
     After rearranging (\ref{mn_2}), we arrive at Eq. (\ref{eq18}).
    \begin{figure*} [b] 
     \hrulefill
        \begin{multline}\label{eq18}
            \tilde{u}(\gamma)_{ij}=\lim _{N \rightarrow \infty} \exp\left (-\theta_{ij}\gamma\right ){\color{orange}\frac{N_{ij}}{(\tilde{N})^{\theta_{ij}}}}{\color{red}\exp \left[2\theta_{ij} \left(\sqrt{\frac{\kappa_i\mu_j}{\theta_{ij}}} -\sqrt{\frac{\tilde{\kappa}\tilde{\mu}}{\tilde{\theta}}} \right)\sqrt{\log (\tilde{N})}\right]}\\\underbrace{
            {\color{red}\exp \left[2\theta_{ij}\sqrt{\frac{\kappa_i\mu_j}{\theta_{ij}}}\sqrt{\frac{\tilde{\kappa}\tilde{\mu}}{\tilde{\theta}}}\right]}\frac{\sqrt{1+\kappa_i}^{\mu_j-\frac{3}{2}}}{\sqrt{\kappa_i}^{\mu_j-\frac{1}{2}}}\frac{1}{2\sqrt{\pi \mu_j}}  \exp \left(-\kappa_{i}\mu_j\right){\color{blue}\exp \left ( \log \left ( \frac{a_N}{\bar{\gamma}} S \right )^{\frac{2\tilde{\mu}-3}{4}} \right ) } \exp \left ( \theta_{ij}c_2 \right )}_{\text{Term-4}}.
        \end{multline}
    \end{figure*}
    Furthermore, Eq.~(\ref{eq18}) can be reduced to Eq.~(\ref{eqST}) by simplifying Term-4 in Eq.~(\ref{eq18}), where multiple exponential terms are combined into a single exponential expression.
    \begin{figure*}    
        \begin{multline}\label{eqST}
            \tilde{u}(\gamma)_{ij}=\lim _{N \rightarrow \infty} \exp\left (-\theta_{ij}\gamma\right ){\color{orange}\frac{N_{ij}}{(\tilde{N})^{\theta_{ij}}}}{\color{red}\exp \left[2\theta_{ij} \left(\sqrt{\frac{\kappa_i\mu_j}{\theta_{ij}}} -\sqrt{\frac{\tilde{\kappa}\tilde{\mu}}{\tilde{\theta}}} \right)\sqrt{\log (\tilde{N})}\right]}\\\underbrace{
            \exp \left[2\theta_{ij}\sqrt{\frac{\kappa_i\mu_j}{\theta_{ij}}}\sqrt{\frac{\tilde{\kappa}\tilde{\mu}}{\tilde{\theta}}}+\log\left(\frac{\sqrt{1+\kappa_i}^{\mu_j-\frac{3}{2}}}{\sqrt{\kappa_i}^{\mu_j-\frac{1}{2}}}\frac{1}{2\sqrt{\pi \mu_j}}\right)+ \left(-\kappa_{i}\mu_j\right)+ \left ( \log \left ( \frac{a_N}{\bar{\gamma}} S \right )^{\frac{2\tilde{\mu}-3}{4}} \right ) +  \theta_{ij}c_2\right]}_{\text{Term-4}}.
        \end{multline}  
    \end{figure*}
    Considering Term-4, choose constant $c_2=\frac{1}{\tilde{\theta}}\left[ \log\left( \frac{\sqrt{\tilde{\kappa}}^{\tilde{\mu}-\frac{1}{2}}}{\sqrt{1+\tilde{\kappa}}^{\tilde{\mu}-\frac{3}{2}}}\frac{2\sqrt{\pi\tilde{\mu}}}{\left(\frac{a_N}{\bar{\gamma}}\right)^{\frac{2\tilde{\mu}-3}{4}}} \right)-\tilde{\kappa}\tilde{\mu} \right]$ .     
    Final expression for single $(i,j)$ with all substitutions 
    \begin{align}\label{ONEIJ}
        &\tilde{u}(\gamma)_{ij}=\lim _{N \rightarrow \infty} \exp\left (-\theta_{ij}\gamma\right ){\color{orange}\frac{N_{ij}}{(\tilde{N})^{\theta_{ij}}}}\nonumber \\ &{\color{red}\exp \left[2\theta_{ij} \left(\sqrt{\frac{\kappa_i\mu_j}{\theta_{ij}}} -\sqrt{\frac{\tilde{\kappa}\tilde{\mu}}{\tilde{\theta}}} \right)\sqrt{\log (\tilde{N})}\right]}\nonumber \\ &
        \exp \left[2\theta_{ij}\sqrt{\frac{\kappa_i\mu_j}{\theta_{ij}}}\sqrt{\frac{\tilde{\kappa}\tilde{\mu}}{\tilde{\theta}}}+\log\left(\frac{\sqrt{1+\kappa_i}^{\mu_j-\frac{3}{2}}}{\sqrt{\kappa_i}^{\mu_j-\frac{1}{2}}}\frac{1}{2\sqrt{\pi \mu_j}}\right)\right] \nonumber \\ &\exp \left[\left(-\kappa_{i}\mu_j\right)+ \left ( \log \left ( \frac{a_N}{\bar{\gamma}} S \right )^{\frac{2\tilde{\mu}-3}{4}} \right ) \right]\nonumber \\ &
        \exp \left( \theta_{ij} \frac{1}{\tilde{\theta}}\left[ \log\left( \frac{\sqrt{\tilde{\kappa}}^{\tilde{\mu}-\frac{1}{2}}}{\sqrt{1+\tilde{\kappa}}^{\tilde{\mu}-\frac{3}{2}}}\frac{2\sqrt{\pi\tilde{\mu}}}{\left(\frac{a_N}{\bar{\gamma}}\right)^{\frac{2\tilde{\mu}-3}{4}}} \right)-\tilde{\kappa}\tilde{\mu} \right] \right).
    \end{align}    
    
    Finally, for all $(i,j)$
    \begin{align}\label{eq:final}
        &\tilde{u}(\gamma)=\lim _{N \rightarrow \infty} \sum_{(i,j) \rightarrow (1,1)}^{(P \times L) }\exp\left (-\theta_{ij}\gamma\right ){\color{orange}\frac{N_{ij}}{(\tilde{N})^{\theta_{ij}}}}\nonumber \\&{\color{red}\exp \left[2\theta_{ij} \left(\sqrt{\frac{\kappa_i\mu_j}{\theta_{ij}}} -\sqrt{\frac{\tilde{\kappa}\tilde{\mu}}{\tilde{\theta}}} \right)\sqrt{\log (\tilde{N})}\right]}\nonumber \\&
        \exp \left[2\theta_{ij}\sqrt{\frac{\kappa_i\mu_j}{\theta_{ij}}}\sqrt{\frac{\tilde{\kappa}\tilde{\mu}}{\tilde{\theta}}}+\log\left(\frac{\sqrt{1+\kappa_i}^{\mu_j-\frac{3}{2}}}{\sqrt{\kappa_i}^{\mu_j-\frac{1}{2}}}\frac{1}{2\sqrt{\pi \mu_j}}\right)\right]\nonumber \\& \exp \left[ \left(-\kappa_{i}\mu_j\right)+ \left ( \log \left ( \frac{a_N}{\bar{\gamma}} S \right )^{\frac{2\tilde{\mu}-3}{4}} \right ) \right] \nonumber \\&
        \exp \left( \theta_{ij} \frac{1}{\tilde{\theta}}\left[ \log\left( \frac{\sqrt{\tilde{\kappa}}^{\tilde{\mu}-\frac{1}{2}}}{\sqrt{1+\tilde{\kappa}}^{\tilde{\mu}-\frac{3}{2}}}\frac{2\sqrt{\pi\tilde{\mu}}}{\left(\frac{a_N}{\bar{\gamma}}\right)^{\frac{2\tilde{\mu}-3}{4}}} \right)-\tilde{\kappa}\tilde{\mu} \right] \right).
    \end{align}   
    Rewriting (\ref{eq:final})
    \begin{equation}\label{comp_1}
         \tilde{u}(\gamma)= \sum_{(i,j) \rightarrow (1,1)}^{(P \times L) }\exp\left (-\theta_{ij}\gamma\right ) P_{ij},
    \end{equation}
    where $P_{ij}$ is given by,
    \begin{multline*}
        P_{ij}=\frac{N_{ij}}{(\tilde{N})^{\theta_{ij}}}\exp \left[2\theta_{ij} \left(\sqrt{\frac{\kappa_i\mu_j}{\theta_{ij}}} -\sqrt{\frac{\tilde{\kappa}\tilde{\mu}}{\tilde{\theta}}} \right)\sqrt{\log (\tilde{N})}\right]\\
        \exp \left[2\theta_{ij}\sqrt{\frac{\kappa_i\mu_j}{\theta_{ij}}}\sqrt{\frac{\tilde{\kappa}\tilde{\mu}}{\tilde{\theta}}}+\log\left(\frac{\sqrt{1+\kappa_i}^{\mu_j-\frac{3}{2}}}{\sqrt{\kappa_i}^{\mu_j-\frac{1}{2}}}\frac{1}{2\sqrt{\pi \mu_j}}\right)\right]\\ \exp \left[ \left(-\kappa_{i}\mu_j\right)+ \left ( \log \left ( \frac{a_N}{\bar{\gamma}} S \right )^{\frac{2\tilde{\mu}-3}{4}} \right ) \right] \\
        \exp \left( \theta_{ij} \frac{1}{\tilde{\theta}}\left[ \log\left( \frac{\sqrt{\tilde{\kappa}}^{\tilde{\mu}-\frac{1}{2}}}{\sqrt{1+\tilde{\kappa}}^{\tilde{\mu}-\frac{3}{2}}}\frac{2\sqrt{\pi\tilde{\mu}}}{\left(\frac{a_N}{\bar{\gamma}}\right)^{\frac{2\tilde{\mu}-3}{4}}} \right)-\tilde{\kappa}\tilde{\mu} \right] \right).
    \end{multline*}   
    The following simplification results in a compact expression. The summation in (\ref{comp_1}) is divided into two terms: i) $(\kappa_i=\tilde{\kappa}, \mu_j=\tilde{\mu})$, ii)  $(\kappa_i\neq\tilde{\kappa}, \mu_j\neq\tilde{\mu})$. Note that $(\kappa_i=\tilde{\kappa}, \mu_j=\tilde{\mu})$ term occurs $\tilde{N}$ times in overall $N$ $(i.e. \quad \tilde{N}=\sum_{n=1}^N \mathbb{I}_{(\kappa_n=\tilde{\kappa},\mu_n=\tilde{\mu})} )$. For this case $\theta_{ij}=1$,  $\lim_{N\rightarrow \infty }\exp \left ( \log \left (\frac{a_N}{\bar{\gamma}} S \right )^{\frac{2\tilde{\mu}-3}{4}} \right )\rightarrow (\frac{a_N}{\bar{\gamma}})^{\frac{2\tilde{\mu}-3}{4}}$,  $\lim_{N\rightarrow \infty }  \exp \left[2\theta_{ij} \left(\sqrt{\frac{\kappa_i\mu_j}{\theta_{ij}}} -\sqrt{\frac{\tilde{\kappa}\tilde{\mu}}{\tilde{\theta}}} \right)\sqrt{\log (\tilde{N})}\right] \rightarrow 1$ as $\tilde{N}\rightarrow \infty$  and $\lim_{N\rightarrow \infty } \frac{N_{ij}}{\tilde{N}^{\theta_{ij}}}\rightarrow 1$. Using the above  and simplifying further, we obtain $ \tilde{u}(\gamma)_{(\kappa_i=\tilde{\kappa}, \mu_j=\tilde{\mu})}=\exp(- \gamma)  $.\\
    For $(\kappa_i\neq\tilde{\kappa},\mu_j\neq\tilde{\mu})$ , as $\theta_{ij}>1$, the term $S=\lim_{N\rightarrow \infty } \frac{\log (\tilde{N})-c_0\log (\log (\tilde{N}))+c_{1} \sqrt{\log (\tilde{N})}-c_2}{(\log (\tilde{N}))^{\theta_{ij} }} \rightarrow 0$. Therefore the term $\lim_{N\rightarrow \infty }\exp \left ( \log \left ( \frac{a_N}{\bar{\gamma}} S \right )^{\frac{2\tilde{\mu}-3}{4}} \right )$ converge to zero. 
     Also, $\lim_{N\rightarrow \infty } \frac{N_{ij}}{\tilde{N}^{\theta_{ij}}}\rightarrow 0$. Hence $ \tilde{u}(\gamma)_{(\kappa_i\neq\tilde{\kappa}, \mu_j\neq\tilde{\mu})}=0 $.
    
    Therefore, 
    \begin{equation}\label{comp_2}
        \tilde{u}(\gamma)=\exp(- \gamma) .
    \end{equation}
    Since, $ \tilde{u}(\gamma)=\exp(- \gamma) <\infty$, we are able to satisfy the condition in Lemma \ref{thm_order_stat} for a non degenerate distribution. Therefore, the normalized $k$-th maximum asymptotic distribution is given by
    \begin{equation}\label{ua3}
        \tilde{\phi}_{N-k+1}\left(\gamma\right) = \sum\limits_{m=0}^{k-1} \frac{\tilde{u}^m(\gamma)}{m!} \exp (-\tilde{u}(\gamma)),
    \end{equation}
    where  $ \tilde{u}(\gamma)=\exp(- \gamma) $.
    Note that, the upper incomplete gamma function for integer $k$ \cite{upper_gamma} is given by
    $$
   \Gamma\left ( k,x \right )=\left ( k-1 \right )!e^{-x}\sum_{r=0}^{k-1}\frac{x^{r}}{r!}.
    $$
    Using the $\Gamma \left ( k,x \right )$ and $\Gamma (x)=(x-1)! $ for integer $x$ \cite{gamma}, we can rewrite (\ref{ua3}) as
     \begin{equation}\label{kth_3}
        \tilde{\phi}_{N-k+1}\left(\gamma\right) = F_{\tilde{\gamma} _{(N-k+1:N)}}\left ( \gamma \right )=\frac{\Gamma \left ( k, \exp(- \gamma)\right )}{\Gamma \left ( k \right )}.
    \end{equation}
   Note that now, $\tilde{\phi}_{k}\left(\gamma\right) $ is a non degenerate function as $\frac{\Gamma \left ( k, \exp(- \gamma)\right )}{\Gamma \left ( k \right )}$ is not a one point distribution. 
    
	 
    \end{appendices}

	\bibliographystyle{IEEEtran}
	\bibliography{reference}

\end{document}